\documentclass[11pt]{article}
\usepackage[margin=1in]{geometry}
\usepackage[final]{microtype}
\usepackage{epsfig}
\usepackage{graphics}
\usepackage{latexsym}
\usepackage{amsmath}
\usepackage{amsfonts}
\usepackage{amssymb}
\usepackage{mathrsfs}
\usepackage[dvipsnames]{xcolor}
\usepackage{amsthm}
\usepackage{xspace}
\usepackage{epstopdf}
\usepackage{float}
\usepackage{hyperref}
\usepackage{pgfplots}
\pgfplotsset{compat=1.18}
\usepackage{longtable}

\numberwithin{equation}{section}
\usepackage[ruled,vlined]{algorithm2e}
\SetArgSty{textrm}

\usepackage[english]{babel}
\usepackage[nottoc]{tocbibind}

\usepackage{caption}
\usepackage{subcaption}
\usepackage{pifont}
\usepackage{booktabs}
\usepackage{multirow}

\theoremstyle{plain}
\newtheorem{theorem}{Theorem}

\newtheorem{lemma}{Lemma}

\newtheorem{mechanism}{Mechanism}
\theoremstyle{definition}

\theoremstyle{remark}

\makeatletter
\@addtoreset{equation}{section}
\def\section{\@startsection {section}{1}{\z@}{-3.5ex plus -1ex minus
 -.2ex}{2.3ex plus .2ex}{\large\bf}}
\makeatother

\def\bfm#1{\mbox{\boldmath$#1$}}

\def\0{\bfm 0}

\DeclareMathAlphabet{\mathpzc}{OT1}{pzc}{m}{it}

\newcounter{my}

\newcounter{my2}

\newcounter{my3}

\newcounter{my4}

\newcounter{my5}

\newcounter{my6}

\allowdisplaybreaks

\begin{document}

\title{Improved Randomized Approximations for Strategic Obnoxious Facility Location}
\author{Hau Chan$^{1}$\quad Jianan Lin$^{2}$\quad Chenhao Wang $^{3,4}$\\[0.75em]
$1$ University of Nebraska-Lincoln\\
$2$ Rensselaer Polytechnic Institute\\
$3$ Beijing Normal University-Zhuhai\\
$4$ Beijing Normal-Hong Kong Baptist University
}
\date{}
\maketitle

\begin{abstract}
    We study randomized strategyproof mechanisms for strategic obnoxious facility location on a line segment, where agents wish the facility to be located as far away from them as possible and their utility is their distance from the facility, under the social utility and minimum utility objectives. For social utility, we propose a novel randomized mechanism that breaks the previously best known \(\frac32\)-approximation of [Cheng, Yu, and Zhang, TCS 2013], achieving an approximation ratio of at most \(1.47359\). We also raise the lower bound on the approximation ratio of randomized strategyproof mechanisms from \(\frac{2}{\sqrt{3}}\approx1.15470\) [Feigenbaum et al., JAAMAS 2020] to \(\frac{105}{88}\approx1.19318\). For minimum utility, following the profile-independent approach of [Chan, Lin and Wang, AAMAS 2026], we design a simple randomized mechanism that reduces the approximation guarantee from \(\sqrt{2n}+O(1)\) to \(\sqrt n+O(1)\), where \(n\) is the number of agents. Finally, we prove that no randomized strategyproof mechanism can achieve an asymptotic approximation ratio strictly smaller than \(2\), strengthening the previous asymptotic lower bound of \(\frac32\) [Feigenbaum et al., JAAMAS 2020].
    Thus, all four bounds considered in this paper strictly improve upon the corresponding previously known results.
\end{abstract}


\section{Introduction}\label{sec:intro}

Strategic facility location is a classical problem that has been actively studied from the mechanism design perspective over the last few decades, with connections to operations research~\cite{alon2010strategyproof,feigenbaum2017approximately}, computer science~\cite{procaccia2013approximate,gravin2025approximation}, and economics~\cite{aziz2020capacity,d1979hotelling}. The standard version of the problem seeks to determine facility locations (e.g., parks, libraries, and schools) in a metric space based on the reported locations of a set of strategic agents, who prefer facilities closer to their ideal locations and may misreport their locations to manipulate the outcome. To address the problem from the mechanism design perspective, existing studies focus on designing strategyproof mechanisms that incentivize agents to report their locations truthfully, while determining facility locations that approximately optimize objectives measuring the cost/utility between agents' ideal locations and the chosen facility locations. These studies include mechanisms for locating one or more facilities in different metric spaces and under various cost/utility objectives \cite{lu10mechanism,DBLP:conf/sagt/Meir19,walsh24utility}.

While the majority of strategic facility location focuses on locating
pleasant facilities that agents prefer to be closer to, a parallel line of
studies considers strategic obnoxious facility location, in which agents
prefer an undesirable facility (e.g., a nuclear plant, a landfill, or a
prison) to be far away from their own locations.
Cheng et al.~\cite{cheng2013strategy}, whose preliminary work appeared
at COCOA 2011, initiated the study of approximation mechanisms for
strategic obnoxious facility location on a path. They examined deterministic
and randomized strategyproof mechanisms for maximizing social utility,
i.e., the sum of agents' distances from the facility.
For deterministic mechanisms, they showed that a simple voting mechanism
between the two endpoints of the segment gives a \(3\)-approximation.
For randomized mechanisms, Cheng et al.~\cite{cheng2013strategy} proposed
an endpoint randomized mechanism and proved that it achieves a
\(\frac32\)-approximation for social utility.

Subsequently, Ibara and Nagamochi~\cite{Ibara2012characterize} studied
strategic obnoxious facility location on a bounded line segment,
which can be normalized to the unit interval \([0,1]\), and characterized
deterministic strategyproof mechanisms independently of any objective.
They showed that every such mechanism has at most two possible candidate
locations across all profiles; that is, its range has size at most two.
Moreover, under anonymity, such a mechanism can be described as a
two-candidate threshold mechanism, where the facility is selected from two
candidate locations according to a threshold rule. In particular,
fixed-location mechanisms appear as a special case.
Together with the characterization of Ibara and
Nagamochi~\cite{Ibara2012characterize}, the results of Cheng et al.
imply that no deterministic strategyproof mechanism can achieve an
approximation ratio better than \(3\), so the deterministic bound is tight.
Ye et al.~\cite{ye2015strategy} studied both the social utility objective
and the objective of sum of squared utilities.
For the former, they proved a lower bound of \(1.077\) on the approximation
ratio of any randomized strategyproof mechanism.

Feigenbaum et al.~\cite{feigenbaum2020strategic} considered hybrid strategic facility location with both agents who prefer to be close to the facility and agents who prefer to be far away from the facility under the social utility and minimum utility objectives. 
In strategic obnoxious facility location (i.e., their hybrid strategic facility location with one type of agents), their results improved the randomized lower bound to
\(\frac{2}{\sqrt{3}}\approx1.15470\) under the social utility objective. 
Under the minimum utility objective, they showed that no deterministic strategyproof mechanism can achieve a bounded approximation ratio by combining the objective with the characterization of \cite{Ibara2012characterize} and established a lower bound of \(\frac32\) for randomized strategyproof mechanisms. 
More recently, Chan et al. \cite{chan2025obnoxious} studied the \(L_p\) utility/cost objectives for strategic obnoxious facility location. 
For minimum utility (i.e., $L_p$ utility with $p=-\infty$), they analyzed the uniform mechanism and obtained an upper bound of \(\sqrt{2n}+O(1)\) where $n$ denotes the number of agents.
Motivated by these gaps, we revisit randomized strategyproof mechanisms for obnoxious facility location on a segment under the social utility and minimum utility objectives. 


\subsection{Our Results}

We focus on randomized strategyproof mechanisms for the obnoxious facility
location on the unit segment \([0,1]\). Agents are distance
maximizers, and we study two utility-maximization objectives: social
utility, which is the sum of agents' distances to the facility, and
minimum utility, which is the minimum distance from any agent to the
facility. Our goal is to improve the best known approximation ratios
and lower bounds for randomized strategyproof mechanisms under these two
objectives.

Our results are summarized in Table~\ref{tab:our-results}. For social
utility, we design a universally strategyproof randomized mechanism that breaks the long-standing
\(\frac32\)-approximation of
\cite{cheng2013strategy}, achieving an approximation ratio of at most
\(1.47359\).
We also raise the randomized lower bound from
\(\frac{2}{\sqrt{3}}\approx1.15470\)~\cite{feigenbaum2020strategic} to
\(\frac{105}{88}\approx1.19318\). For minimum utility, we refine the
uniform mechanism analyzed in \cite{chan2025obnoxious}, which selects a facility uniformly at random from
the whole segment, and we obtain an
upper bound of \(\sqrt n+O(1)\), reducing the leading coefficient from \(\sqrt2\) to \(1\). This mechanism is also universally strategyproof. Finally, we strengthen the lower bound for
randomized strategyproof mechanisms from
\(\frac32\)~\cite{feigenbaum2020strategic} to an asymptotic lower bound of \(2\). Thus, all four bounds considered in this paper strictly improve
the corresponding previously known results.

\begin{table}[t]
    \centering
    \caption{Summary of our improved bounds for randomized strategyproof
    mechanisms on a line segment. Upper bounds are approximation ratios
    achieved by mechanisms, while lower bounds are impossibility results.}
    \label{tab:our-results}
    \renewcommand{\arraystretch}{1.15}
    \begin{tabular}{cccc}
        \toprule
        Objective & Bound & Previous best & This paper \\
        \midrule
        Social utility
        & Upper bound
        & \(\frac32\)~\cite{cheng2013strategy}
        & \(1.47359\) \\
        Social utility
        & Lower bound
        & \(\frac{2}{\sqrt{3}}\approx1.15470\)~\cite{feigenbaum2020strategic}
        & \(\frac{105}{88}\approx1.19318\) \\
        Minimum utility
        & Upper bound
        & \(\sqrt{2n}+O(1)\)~\cite{chan2025obnoxious}
        & \(\sqrt n+O(1)\) \\
        Minimum utility
        & Lower bound
        & \(\frac32\) asymptotically~\cite{feigenbaum2020strategic}
        & \(2\) asymptotically \\
        \bottomrule
    \end{tabular}
\end{table}

\subsection{Additional Related Work}

We have focused so far on the works most directly related to our results
on obnoxious facility location on a segment. Beyond the line-segment
setting, Oomine and Nagamochi~\cite{oomine2016characterizing} studied
strategyproof mechanisms for locating an obnoxious facility on tree
networks. Recent variants of obnoxious facility location also incorporate
fairness requirements~\cite{alex2024}, group-fairness constraints
\cite{li2024strategyproof}, and predictions into the mechanism-design
model~\cite{DBLP:journals/corr/abs-2212-09521}. These works illustrate
that, although the obnoxious facility location problem is structurally
simple, the interaction between strategic behavior and different
performance or fairness objectives remains subtle.

Our work is also related to the broader literature on strategyproof
facility location, where agents prefer the facility to be close rather
than far away. Moulin~\cite{moulin1980strategy} gave the classical
characterization of strategyproof mechanisms on the line via generalized
median mechanisms. Procaccia and Tennenholtz~\cite{procaccia2013approximate}
initiated the study of approximation ratios for strategyproof facility
location mechanisms without money.
Other extensions include two-facility location~\cite{lu2009tighter,lu10mechanism,chan2026randomized}, network settings (e.g., trees and circles)~\cite{alon2010strategyproof,DBLP:conf/sagt/Meir19,rogowski2025improved}, and multi-dimensional spaces~\cite{tang2020characterization,lin2020nearly,GoelH23,gravin2025approximation,barak2026facility,chan2026strategyproof}.

\section{Preliminaries}\label{sec:model}

We study the obnoxious facility location problem on a unit-length line segment. 
The feasible region is the interval $[0,1]$. Let $N=[n]=\{1,2,\ldots,n\}$ be the set of agents. Each agent $i\in N$ has a private location $x_i\in[0,1]$. A location profile is denoted by $\mathbf{x}=(x_1,\ldots,x_n)\in[0,1]^n$, and we write $\mathbf{x}_{-i}$ for the profile obtained by removing agent $i$.

A facility location is a point $y\in[0,1]$. Since the facility is obnoxious, agents prefer the facility to be as far away from them as possible. Thus, for an agent located at $x_i$ and a facility located at $y$, the utility of agent $i$ is \(u(x_i,y)=|x_i-y|\).

Throughout the paper, we focus on randomized mechanisms. A randomized mechanism $f$ maps every reported profile $\mathbf{x}\in[0,1]^n$ to a probability distribution over $[0,1]$. We use $Y\sim f(\mathbf{x})$ to denote the random facility location sampled from this distribution. For a randomized mechanism, the utility of agent $i$ is its expected distance from the random facility location: \(u(x_i,f(\mathbf{x}))
    =
    \mathbb{E}_{Y\sim f(\mathbf{x})}\bigl[|x_i-Y|\bigr]\).

A mechanism $f$ is \emph{strategyproof} (SP) if no agent can increase her expected utility by misreporting her location. Formally, for every agent $i\in N$, every profile $\mathbf{x}\in[0,1]^n$, and every possible misreport $x_i'\in[0,1]$,
\[
    \mathbb{E}_{Y\sim f(x_i,\mathbf{x}_{-i})}\bigl[|x_i-Y|\bigr]
    \ge
    \mathbb{E}_{Y\sim f(x_i',\mathbf{x}_{-i})}\bigl[|x_i-Y|\bigr].
\]
While this definition of strategyproofness is based on expectation, there is a stronger notion: a randomized mechanism is called \emph{universally strategyproof} if it is a distribution over deterministic strategyproof mechanisms.

We consider two utility-maximization objectives. For a deterministic facility location $y\in[0,1]$, we define the \emph{social utility} and \emph{minimum utility} as
\[
    \operatorname{SU}(\mathbf{x},y)
    =
    \sum_{i\in N}|x_i-y|,
    \qquad\text{and}\qquad
    \operatorname{MU}(\mathbf{x},y)
    =
    \min_{i\in N}|x_i-y|.
\]

For a randomized mechanism $f$, the objective value is defined as the expected objective value with respect to the random facility location $Y\sim f(\mathbf{x})$. That is,
\[
    \operatorname{SU}(\mathbf{x},f(\mathbf{x}))
    =
    \mathbb{E}_{Y\sim f(\mathbf{x})}
    \left[
        \sum_{i\in N}|x_i-Y|
    \right],
\qquad\text{and}\qquad
    \operatorname{MU}(\mathbf{x},f(\mathbf{x}))
    =
    \mathbb{E}_{Y\sim f(\mathbf{x})}
    \left[
        \min_{i\in N}|x_i-Y|
    \right].
\]

For $T\in\{\operatorname{SU},\operatorname{MU}\}$, let \(\operatorname{OPT}_T(\mathbf{x})
    =
    \max_{y\in[0,1]} T(\mathbf{x},y)\)
be the optimal value of objective $T$ on profile $\mathbf{x}$. A randomized mechanism $f$ is an $\alpha$-approximation for objective $T\in\{\operatorname{SU},\operatorname{MU}\}$ if for every profile $\mathbf{x}\in[0,1]^n$,
\[
    \operatorname{OPT}_T(\mathbf{x})
    \le
    \alpha\cdot T(\mathbf{x},f(\mathbf{x})).
\]
By definition, an approximation ratio is at least $1$.

\section{Social Utility}

Throughout this section, write \(\operatorname{OPT}=\operatorname{OPT}_{\operatorname{SU}}\).

In this section, we study social utility for both upper and lower bounds.

\subsection{Upper Bound}

Before defining our mechanism,
we first recall the classical strategyproof randomized mechanism of Cheng et
al.~\cite{cheng2013strategy}, which is the known \(\frac32\)-approximation
mechanism for the obnoxious facility location problem on an interval.

\begin{mechanism}\label{mec:cheng}
(Endpoint Randomized Mechanism~\cite{cheng2013strategy}).
Given a reported profile \(\hat{\mathbf{x}}\), let
\(N_L=\{i\in N:\hat{x}_i\le \frac12\}\),
\(N_R=\{i\in N:\hat{x}_i>\frac12\}\). Denote
\(\ell=|N_L|\) and \(r=|N_R|\). The mechanism outputs \(0\) with
probability \(\frac{r^2+2\ell r}{\ell^2+r^2+4\ell r}\), 
and outputs \(1\) with the remaining probability \(\frac{\ell^2+2\ell r}{\ell^2+r^2+4\ell r}\).
\end{mechanism}

The mechanism above only places the facility at one of the two endpoints.
In contrast, our mechanism uses a random threshold and may also place the
facility at interior points. We denote our mechanism by
\(\mathcal{M}_\lambda\), where \(\lambda\in[0,1]\) is a parameter. In the
main analysis, we will set \(\lambda=\frac{93}{100}\).

\begin{mechanism}\label{mec:ran}
(Random Threshold Dictator \(\mathcal{M}_\lambda\)).
Given a reported profile \(\hat{\mathbf{x}}\), draw a threshold \(T\) as
follows: with probability \(\lambda\), draw \(T\) uniformly from
\([0,1]\); with probability \(1-\lambda\), set \(T=\frac12\).
Independently, draw an agent \(J\) uniformly at random from \(N\). Define
\(a_T=\max\{0,2T-1\}\) and \(b_T=\min\{1,2T\}\). The mechanism outputs
\(b_T\) if \(\hat{x}_J\le T\), and outputs \(a_T\) otherwise.
\end{mechanism}

For every threshold \(T\), the two candidate locations \(a_T\) and \(b_T\)
are symmetric around \(T\), subject to staying inside the interval
\([0,1]\). More precisely, they are the endpoints of the largest
subinterval of \([0,1]\) whose midpoint is \(T\). Therefore
\(a_T\le T\le b_T\) and \(a_T+b_T=2T\). After \(T\) and \(J\) are chosen,
the selected agent decides which of these two candidates is used: if the
selected agent reports weakly to the left of \(T\), the mechanism chooses
the right candidate \(b_T\); otherwise, it chooses the left candidate
\(a_T\).

\begin{theorem}\label{thm:su-upper}
The mechanism \(\mathcal{M}_{93/100}\) is universally strategyproof and
achieves an approximation ratio for social utility of at most
\[
    \alpha_0
    =
    \left(
        2\sqrt{
            \frac{7641}{10000}\cdot\frac{6869}{5000}
        }
        -
        \frac{2741}{2000}
    \right)^{-1}
    <
    1.47359.
\]
\end{theorem}

We first show its strategyproofness.

\begin{lemma}\label{lem:sp}
For every \(\lambda\in[0,1]\), the mechanism \(\mathcal{M}_\lambda\) is
strategyproof. In fact, it is universally strategyproof.
\end{lemma}

\begin{proof}
Fix an arbitrary realization of the random choices, namely a threshold
\(T\) and a selected agent \(J\). 
Consider first an agent \(i\ne J\). The report of agent \(i\) does not
affect the outcome, and hence agent \(i\) cannot benefit by misreporting.
It remains to consider the selected agent \(J\).
Let the true location of
this agent be \(x_J\). The only two possible outcomes are \(a_T\) and
\(b_T\), whose midpoint is \(T\).

If \(x_J\le T\), then \(b_T\) is weakly farther from \(x_J\) than \(a_T\).
Indeed,
\[
    |x_J-b_T|^2-|x_J-a_T|^2
    =
    (b_T-a_T)(a_T+b_T-2x_J)
    =
    2(b_T-a_T)(T-x_J)
    \ge 0.
\]
Thus truthful reporting makes the mechanism choose \(b_T\), which is a
weakly preferred outcome for agent \(J\). Any report on the same side of
\(T\) leaves the outcome unchanged, while any report on the other side
changes the outcome to \(a_T\), which is no better.

The case \(x_J>T\) is symmetric. In this case \(a_T\) is weakly farther
from \(x_J\) than \(b_T\), since
\[
    |x_J-a_T|^2-|x_J-b_T|^2
    =
    2(b_T-a_T)(x_J-T)
    \ge 0.
\]
Therefore truthful reporting again selects a weakly preferred outcome, and
no misreport can improve the selected agent's utility.

Hence, for every fixed realization \((T,J)\), truthful reporting is a
dominant strategy for every agent. Since the distribution of \((T,J)\) is
independent of the reported profile, \(\mathcal{M}_\lambda\) is
strategyproof. Moreover, because every deterministic mechanism obtained by
fixing the random choices is strategyproof, \(\mathcal{M}_\lambda\) is
universally strategyproof.
\end{proof}

We now turn to the approximation ratio.
The main difficulty is to
lower bound the expected social utility of \(\mathcal{M}_\lambda\) for
every profile, even though the mechanism may output both endpoints and
interior points. Our proof uses the random-dictator structure of the
mechanism to rewrite its expected social utility as a two-variable kernel
expression. We then prove a pointwise quadratic lower bound on this kernel,
which reduces the approximation analysis to a one-dimensional optimization
over the average location of the agents.

We begin with the kernel representation of the expected social utility. For
\(x,z\in[0,1]\), define
\begin{align}
    K(x,z)
    &=
    \int_0^z |x-a_T|\,dT
    +
    \int_z^1 |x-b_T|\,dT,
    \label{eq:uniform-kernel}\\
    H(x,z)
    &=
    \begin{cases}
        1-x, & z\le \frac12,\\
        x, & z>\frac12.
    \end{cases}
    \label{eq:atom-kernel}
\end{align}
Here \(K(x,z)\) is the expected utility of an agent at \(x\) when the
selected agent is at \(z\) and the threshold is drawn uniformly from
\([0,1]\). Similarly, \(H(x,z)\) is the utility contribution from the
case where the threshold is fixed at \(\frac12\). We will explain this in the proof of Lemma~\ref{lem:kernel-representation}. For
\(\lambda\in[0,1]\), let
\begin{align}
    K_\lambda(x,z)
    &=
    \lambda K(x,z)+(1-\lambda)H(x,z),
    \label{eq:mixture-kernel}\\
    K_\lambda^s(x,z)
    &=
    \frac{K_\lambda(x,z)+K_\lambda(z,x)}{2}.
    \label{eq:symmetric-kernel}
\end{align}

\begin{lemma}\label{lem:kernel-representation}
Let \(\mathbf{x}\in[0,1]^n\) be any profile, and let \(X\) and \(Z\) be
two independent samples drawn uniformly from the multiset
\(\{x_1,\ldots,x_n\}\). Then
\begin{align}
    \frac{1}{n}
    \operatorname{SU}(\mathbf{x},\mathcal{M}_\lambda(\mathbf{x}))
    =
    \mathbb{E}\left[K_\lambda^s(X,Z)\right].
    \label{eq:kernel-representation}
\end{align}
\end{lemma}

\begin{proof}
Fix two agents \(i,j\), and suppose that agent \(j\) is selected by the
mechanism. Write \(x=x_i\) and \(z=x_j\). If the threshold is drawn
uniformly from \([0,1]\), then for \(T<z\) the selected agent reports to
the right of \(T\), so the mechanism outputs \(a_T\); for \(T\ge z\), it
outputs \(b_T\). The value of the outcome at the single point \(T=z\)
does not affect the integral. Hence the expected utility of agent \(i\)
from the uniform-threshold part is exactly \(K(x,z)\), as defined in
\eqref{eq:uniform-kernel}.

If the threshold is fixed at \(T=\frac12\), then \(a_T=0\) and \(b_T=1\).
When \(z\le\frac12\), the mechanism outputs \(1\), giving utility
\(1-x\) to an agent at \(x\). When \(z>\frac12\), the mechanism outputs
\(0\), giving utility \(x\). This is exactly \(H(x,z)\), as defined in
\eqref{eq:atom-kernel}. Therefore, conditional on selecting agent \(j\),
the expected utility of agent \(i\) is \(K_\lambda(x_i,x_j)\).

Averaging over the uniformly selected agent \(j\), and summing over all
agents \(i\), we obtain
\begin{align*}
    \operatorname{SU}(\mathbf{x},\mathcal{M}_\lambda(\mathbf{x}))
    =
    \frac{1}{n}
    \sum_{j=1}^n\sum_{i=1}^n K_\lambda(x_i,x_j) = n\cdot \frac{1}{n^2}
    \sum_{i=1}^n\sum_{j=1}^n K_\lambda(x_i,x_j) = n \mathbb{E}[K_\lambda(X,Z)].
\end{align*}
Since \(X\) and \(Z\) are independent and identically distributed,
\(\mathbb{E}[K_\lambda(X,Z)]=\mathbb{E}[K_\lambda(Z,X)]
=\mathbb{E}[K_\lambda^s(X,Z)]\).
Dividing the welfare identity by \(n\) and using this equality proves
\eqref{eq:kernel-representation}.
\end{proof}

We next record explicit formulas for the kernels. These formulas will be
used to prove the pointwise quadratic lower bound.

\begin{lemma}\label{lem:kernel-formulas}
For every \(x,z\in[0,1]\),
\begin{align}
    K(x,z)
    =
    \begin{cases}
        \frac12 x^2+z^2-x+\frac34,
        & x\ge 2z,\\[3pt]
        2xz-x-z^2+\frac34,
        & 2z-1\le x\le 2z,\\[3pt]
        \frac12 x^2+z^2-2z+\frac54,
        & x\le 2z-1.
    \end{cases}
    \label{eq:uniform-kernel-formula}
\end{align}
Moreover, if \(0\le x\le z\le1\), then the symmetrized atom kernel
\(H^s(x,z)=\frac{H(x,z)+H(z,x)}{2}\) is given by
\begin{align}
    H^s(x,z)
    =
    \begin{cases}
        1-\frac{x+z}{2},
        & z\le \frac12,\\[3pt]
        \frac{x+1-z}{2},
        & x\le \frac12<z,\\[3pt]
        \frac{x+z}{2},
        & \frac12<x.
    \end{cases}
    \label{eq:atom-kernel-formula}
\end{align}
\end{lemma}

\begin{proof}
We first derive \eqref{eq:uniform-kernel-formula}. Recall that
\(a_T=\max\{0,2T-1\}\) and \(b_T=\min\{1,2T\}\). In the first integral of
\eqref{eq:uniform-kernel}, namely \(\int_0^z |x-a_T|\,dT\), the sign of
\(x-a_T\) can change only when \(a_T=x\). Since \(a_T=0\) for
\(T\le \frac12\) and \(a_T=2T-1\) for \(T\ge \frac12\), and since
\(x\in[0,1]\), this crossing occurs at \(T=\frac{x+1}{2}\). Similarly,
in the second integral \(\int_z^1 |x-b_T|\,dT\), the sign of \(x-b_T\)
can change only when \(b_T=x\). Since \(b_T=2T\) for
\(T\le \frac12\) and \(b_T=1\) for \(T\ge \frac12\), this crossing occurs
at \(T=\frac{x}{2}\). 

Therefore, the behavior of the absolute values in
\eqref{eq:uniform-kernel} is completely determined by whether the interval
\([0,z]\) contains the crossing point \(\frac{x+1}{2}\), and whether the
interval \([z,1]\) contains the crossing point \(\frac{x}{2}\). Equivalently,
the relevant comparisons are \(x\le 2z-1\) and \(x\le 2z\).

If \(2z-1\le x\le 2z\), then this is exactly the case in which
\(\frac{x+1}{2}\ge z\) and \(\frac{x}{2}\le z\), so no sign change occurs
inside either integral. Thus \(a_T\le x\) for all \(T\le z\), and
\(b_T\ge x\) for all \(T\ge z\). Hence
\[
    K(x,z)
    =
    \int_0^z (x-a_T)\,dT
    +
    \int_z^1 (b_T-x)\,dT .
\]
We compute this expression by considering whether \(z\le\frac12\) or
\(z\ge\frac12\). If \(z\le\frac12\), then \(a_T=0\) on \([0,z]\),
\(b_T=2T\) on \([z,\frac12]\), and \(b_T=1\) on
\([\frac12,1]\). Therefore
\begin{align*}
    K(x,z)
    &=
    \int_0^z x\,dT
    +
    \int_z^{1/2}(2T-x)\,dT
    +
    \int_{1/2}^1(1-x)\,dT  =
    xz
    +
    \left[T^2-xT\right]\bigg|_z^{1/2}
    +
    \frac{1-x}{2}  \\
    &=
    xz
    +
    \left(\frac14-\frac{x}{2}-z^2+xz\right)
    +
    \frac{1-x}{2}  =
    2xz-x-z^2+\frac34 .
\end{align*}
If \(z\ge\frac12\), then \(a_T=0\) on \([0,\frac12]\),
\(a_T=2T-1\) on \([\frac12,z]\), and \(b_T=1\) on \([z,1]\).
Therefore
\begin{align*}
    K(x,z)
    &=
    \int_0^{1/2}x\,dT
    +
    \int_{1/2}^z(x-(2T-1))\,dT
    +
    \int_z^1(1-x)\,dT  \\
    &=
    \frac{x}{2}
    +
    \left[xT-T^2+T\right]\bigg|_{1/2}^z
    +
    (1-z)(1-x)  \\
    &=
    \frac{x}{2}
    +
    \left(xz-z^2+z-\frac{x}{2}-\frac14\right)
    +
    (1-z)(1-x) =
    2xz-x-z^2+\frac34 .
\end{align*}
Thus, in both subcases,
\[
    K(x,z)=2xz-x-z^2+\frac34 .
\]

If \(x\ge 2z\), then \(b_T\le x\) for
\(T\in[z,\frac{x}{2}]\), and \(b_T\ge x\) for
\(T\in[\frac{x}{2},1]\). Compared with the previous case, only the second
integral changes sign on \(T\in[z,\frac{x}{2}]\). Therefore
\begin{align*}
    K(x,z)
    &=
    2xz-x-z^2+\frac34
    +
    2\int_z^{x/2}(x-b_T)\,dT .
\end{align*}
On \(T\in[z,\frac{x}{2}]\), we have \(T\le\frac{x}{2}\le\frac12\), so
\(b_T=2T\). Thus
\begin{align*}
    2\int_z^{x/2}(x-b_T)\,dT
    &=
    2\int_z^{x/2}(x-2T)\,dT  =
    2\left[xT-T^2\right]\bigg|_{z}^{x/2} =
    2\left(\frac{x^2}{4}-xz+z^2\right)  \\
    &=
    \frac{x^2}{2}-2xz+2z^2 .
\end{align*}
Hence
\[
    K(x,z)
    =
    2xz-x-z^2+\frac34
    +
    \frac{x^2}{2}-2xz+2z^2
    =
    \frac12x^2+z^2-x+\frac34 .
\]

Finally, if \(x\le 2z-1\), then \(a_T\ge x\) for
\(T\in[\frac{x+1}{2},z]\). Compared with the middle case, only the first
integral changes sign on this interval. Therefore
\begin{align*}
    K(x,z)
    &=
    2xz-x-z^2+\frac34
    +
    2\int_{(x+1)/2}^{z}(a_T-x)\,dT .
\end{align*}
On \(T\in[\frac{x+1}{2},z]\), we have \(T\ge\frac{x+1}{2}\ge\frac12\),
so \(a_T=2T-1\). Thus
\begin{align*}
    2\int_{(x+1)/2}^{z}(a_T-x)\,dT
    &=
    2\int_{(x+1)/2}^{z}(2T-1-x)\,dT =
    2\left[T^2-(1+x)T\right]\bigg|_{(x+1)/2}^{z}  \\
    &=
    2\left(z^2-(1+x)z+\frac{(1+x)^2}{4}\right)  =
    2z^2-2z-2xz+\frac{(1+x)^2}{2}.
\end{align*}
Combining this with the middle-case expression gives
\begin{align*}
    K(x,z)
    &=
    2xz-x-z^2+\frac34
    +
    2z^2-2z-2xz+\frac{(1+x)^2}{2}  =
    \frac12x^2+z^2-2z+\frac54 .
\end{align*}
This proves \eqref{eq:uniform-kernel-formula}.

It remains to prove \eqref{eq:atom-kernel-formula}. When
\(T=\frac12\), the two candidates are \(0\) and \(1\). If the selected
agent is weakly to the left of \(\frac12\), the mechanism outputs \(1\);
otherwise it outputs \(0\). Hence \(H(x,z)=1-x\) when \(z\le\frac12\),
and \(H(x,z)=x\) when \(z>\frac12\). Symmetrizing this expression gives
the three cases in \eqref{eq:atom-kernel-formula}. For the second case, the two terms are as follows:

In \(H(x,z)\), the agent whose utility is measured is
located at \(x\), while the selected agent is located at \(z\). Since
\(z>\frac12\), the mechanism outputs \(0\), and the utility is \(x\).
Thus \(H(x,z)=x\). In \(H(z,x)\), the roles of \(x\) and \(z\) are
reversed: the agent whose utility is measured is located at \(z\), while
the selected agent is located at \(x\). Since \(x\le\frac12\), the
mechanism outputs \(1\), and the utility is \(1-z\).
\end{proof}

The next lemma is the main technical ingredient of the upper-bound proof.
It gives a pointwise quadratic lower bound for the symmetrized kernel when
\(\lambda=\frac{93}{100}\). The form \(A+B(x+z)+Cxz\) is useful because its expectation under independent, identically distributed samples depends only on their common mean. The rational coefficients below admit an exact certificate, which is verified in Appendix~\ref{app:upper-kernel}.

\begin{lemma}\label{lem:quadratic-kernel-lower-bound}
Let \(\lambda=\frac{93}{100}\). For every \(x,z\in[0,1]\),
\begin{align}
    K_{93/100}^s(x,z)
    \ge
    A+B(x+z)+Cxz,
    \qquad
    A=\frac{7641}{10000},
    \quad
    B=-\frac{2741}{4000},
    \quad
    C=\frac{6869}{5000}.
    \label{eq:quadratic-kernel-lower-bound}
\end{align}
\end{lemma}

\begin{proof}
Let
\begin{align}
    P(x,z)=A+B(x+z)+Cxz.
    \label{eq:quadratic-P}
\end{align}
Both \(K_{93/100}^s(x,z)\) and \(P(x,z)\) are symmetric in \(x\) and
\(z\). Hence it suffices to prove the claim for
\(0\le x\le z\le1\).

By Lemma~\ref{lem:kernel-formulas}, the uniform-threshold part of the
kernel is piecewise quadratic. On the triangle \(0\le x\le z\le1\), the
pieces are determined by the two comparisons \(z\le 2x\) and
\(z\le \frac{1+x}{2}\). Define
\begin{align}
    U_{AA}
    &=
    \left\{
        (x,z): 0\le x\le z\le1,\ 
        z\le 2x,\ 
        z\le \frac{1+x}{2}
    \right\}, \notag\\
    U_{AB}
    &=
    \left\{
        (x,z): 0\le x\le z\le1,\ 
        z\ge 2x,\ 
        z\le \frac{1+x}{2}
    \right\}, \notag\\
    U_{CA}
    &=
    \left\{
        (x,z): 0\le x\le z\le1,\ 
        z\le 2x,\ 
        z\ge \frac{1+x}{2}
    \right\}, \notag\\
    U_{CB}
    &=
    \left\{
        (x,z): 0\le x\le z\le1,\ 
        z\ge 2x,\ 
        z\ge \frac{1+x}{2}
    \right\}.
    \label{eq:uniform-regions}
\end{align}
The atom part, corresponding to the event \(T=\frac12\), is determined by
the relative positions of \(x,z\) and \(\frac12\). Define
\begin{align}
    V_L
    &=
    \left\{(x,z):0\le x\le z\le \frac12\right\}, \notag\\
    V_X
    &=
    \left\{(x,z):0\le x\le \frac12<z\le1\right\}, \notag\\
    V_R
    &=
    \left\{(x,z):\frac12<x\le z\le1\right\}.
    \label{eq:atom-regions}
\end{align}
The nonempty intersections of the regions in
\eqref{eq:uniform-regions} and \eqref{eq:atom-regions} cover
the triangle \(0\le x\le z\le1\), with overlaps only on uniform-region boundaries where the polynomial formulas agree. On each such
intersection,
\[
    D(x,z)=K_{93/100}^s(x,z)-P(x,z)
\]
is a quadratic polynomial in \(x\) and \(z\).

Table~\ref{tab:su-upper-certificate} in Appendix~\ref{app:upper-kernel} gives an exact certificate
for the nonnegativity of \(D\). For each nonempty intersection \(S\),
let \(Q_S\) be the quadratic polynomial equal to \(10000D\) on \(S\).
The table minimizes the continuous polynomial \(Q_S\) over the closure
\(\overline S\), by checking its stationary points and boundary segments.
Since every listed minimum is nonnegative, \(D=Q_S/10000\ge0\) on \(S\).
This argument uses the continuity of each polynomial, not of \(D\),
which can jump at the atom-region boundaries. The nonempty intersections
cover the triangle, with boundary values assigned according to
\eqref{eq:atom-regions}, so \(D(x,z)\ge0\) everywhere on it. Thus
\(K_{93/100}^s(x,z)\ge P(x,z)\) for all \(0\le x\le z\le1\), and by
symmetry for all \(x,z\in[0,1]\). This proves
\eqref{eq:quadratic-kernel-lower-bound}.
\end{proof}

We now use the quadratic kernel lower bound to derive the approximation
guarantee.

\begin{lemma}\label{lem:approximation-from-kernel}
Let \(\lambda=\frac{93}{100}\), and let \(A,B,C\) be the constants in
\eqref{eq:quadratic-kernel-lower-bound}. For every profile \(\mathbf{x}\in[0,1]^n\),
\begin{align}
    \operatorname{SU}(\mathbf{x},\mathcal{M}_{93/100}(\mathbf{x}))
    \ge
    r_0\cdot \operatorname{OPT}(\mathbf{x}),
    \label{eq:approximation-guarantee}
\end{align}
where
\begin{align}
    r_0
    =
    2\sqrt{AC}+2B
    =
    2\sqrt{
        \frac{7641}{10000}\cdot\frac{6869}{5000}
    }
    -
    \frac{2741}{2000}.
    \label{eq:r0-definition}
\end{align}
\end{lemma}

\begin{proof}
Let \(\bar{x}=\frac1n\sum_{i=1}^n x_i\). We first consider the case
\(\bar{x}\ge \frac12\). By Lemma~\ref{lem:kernel-representation} and
Lemma~\ref{lem:quadratic-kernel-lower-bound},
\begin{align}
    \frac{1}{n}
    \operatorname{SU}(\mathbf{x},\mathcal{M}_{93/100}(\mathbf{x}))
    &=
    \mathbb{E}\left[K_{93/100}^s(X,Z)\right] \ge
    \mathbb{E}\left[A+B(X+Z)+CXZ\right] \notag\\
    &=
    A+2B\bar{x}+C\bar{x}^2.
    \label{eq:expected-su-average-bound}
\end{align}
Here \(X\) and \(Z\) are independent samples drawn uniformly from the
multiset \(\{x_1,\ldots,x_n\}\). Since \(\bar{x}\ge\frac12\), the optimal
solution places the facility at \(0\), and hence
\(\operatorname{OPT}(\mathbf{x})=n\bar{x}\). Therefore
\begin{align}
    \frac{
        \operatorname{SU}(\mathbf{x},\mathcal{M}_{93/100}(\mathbf{x}))
    }{
        \operatorname{OPT}(\mathbf{x})
    }
    \ge
    \frac{A+2B\bar{x}+C\bar{x}^2}{\bar{x}}.
    \label{eq:ratio-as-function}
\end{align}
For \(t\in[\frac12,1]\), define
\[
    f(t)=\frac{A+2Bt+Ct^2}{t}=\frac{A}{t}+2B+Ct.
\]
The function \(f\) is minimized at \(t=\sqrt{A/C}\), which lies in
\([\frac12,1]\). Thus
\begin{align}
    f(t)
    \ge
    2\sqrt{AC}+2B
    =
    r_0
    \qquad
    \text{for all }t\in\left[\frac12,1\right].
    \label{eq:f-minimum}
\end{align}
Combining \eqref{eq:ratio-as-function} and \eqref{eq:f-minimum} proves
\eqref{eq:approximation-guarantee} when \(\bar{x}\ge\frac12\).

It remains to handle the case \(\bar{x}\le\frac12\). Let
\(\mathbf{x}^{R}\) be the reflected profile, defined by
\(x_i^{R}=1-x_i\) for every agent \(i\). Then
\(\frac1n\sum_i x_i^R=1-\bar{x}\ge\frac12\), and
\(\operatorname{OPT}(\mathbf{x}^{R})=\operatorname{OPT}(\mathbf{x})\).

We compare the expected social utility of \(\mathcal{M}_{93/100}\) on
\(\mathbf{x}\) and on \(\mathbf{x}^{R}\). Couple the two executions as
follows. Use the same selected agent \(J\). If the threshold in the
execution on \(\mathbf{x}\) is \(T\), use threshold \(1-T\) in the
execution on \(\mathbf{x}^{R}\). This coupling is valid because the
threshold distribution is invariant under the map \(T\mapsto 1-T\).

For every threshold \(T\), the corresponding candidates satisfy \(a_{1-T}=1-b_T\) and \(b_{1-T}=1-a_T\).
If \(x_J\ne T\), the two coupled outcomes are reflections of each other.
Therefore their social utilities are equal. The only event on which the
tie-breaking can matter has positive probability only when
\(T=\frac12\) and \(x_J=\frac12\). On this event, both executions output
\(1\). The social utility on \(\mathbf{x}\) is then
\(n-\sum_i x_i=n(1-\bar{x})\), whereas the social utility on
\(\mathbf{x}^{R}\) is \(\sum_i x_i=n\bar{x}\). Since
\(\bar{x}\le\frac12\), the former is at least the latter. Hence
\begin{align}
    \operatorname{SU}(\mathbf{x},\mathcal{M}_{93/100}(\mathbf{x}))
    \ge
    \operatorname{SU}(\mathbf{x}^{R},\mathcal{M}_{93/100}(\mathbf{x}^{R}))
    \ge
    r_0\cdot \operatorname{OPT}(\mathbf{x}^{R})
    =
    r_0\cdot \operatorname{OPT}(\mathbf{x}). \notag
\end{align}
This proves \eqref{eq:approximation-guarantee} for \(\bar{x}\le\frac12\) as well.
\end{proof}

We can now prove the main theorem.

\begin{proof}[Proof of Theorem~\ref{thm:su-upper}]
By Lemma~\ref{lem:sp}, \(\mathcal{M}_{93/100}\) is universally
strategyproof. By Lemma~\ref{lem:approximation-from-kernel},
\[
    \operatorname{SU}(\mathbf{x},\mathcal{M}_{93/100}(\mathbf{x}))
    \ge
    r_0\cdot \operatorname{OPT}(\mathbf{x})
\]
for every profile \(\mathbf{x}\), where \(r_0\) is defined in
\eqref{eq:r0-definition}. Therefore the approximation ratio of
\(\mathcal{M}_{93/100}\) is at most
\[
    \frac{1}{r_0}
    =
    \left(
        2\sqrt{
            \frac{7641}{10000}\cdot\frac{6869}{5000}
        }
        -
        \frac{2741}{2000}
    \right)^{-1}
    \approx
    1.473584273
    <
    1.47359.
\]
This is strictly smaller than \(\frac32\).
\end{proof}

\subsection{Lower Bound}

We now turn to lower bounds. Feigenbaum et al.
\cite{feigenbaum2020strategic} proved a lower bound of
\(\frac{2}{\sqrt{3}}\approx 1.15470\) for randomized strategyproof
mechanisms. We give an explicit lower bound that improves this to
\(\frac{105}{88}\approx 1.19318\). The improvement is approximately \(0.03848\).

\begin{theorem}\label{thm:su-lower}
Every randomized strategyproof mechanism has approximation ratio at least \(\frac{105}{88} \approx 1.19318\) for social utility.
\end{theorem}

\begin{proof}
By scaling the interval, it suffices to prove the lower bound on
\([0,8]\). Consider two agents. For a profile \((a,b)\), let \(Y_{ab}\)
denote the random facility location selected by the mechanism on that
profile, and define \(D_{ab}(t)=\mathbb{E}[|t-Y_{ab}|]\).
Also let \(W_{ab}=D_{ab}(a)+D_{ab}(b)\) be the expected social utility on
profile \((a,b)\). On the interval \([0,8]\), the optimal social utility
for profile \((a,b)\) is \(\operatorname{OPT}(a,b)=\max\{a+b,16-a-b\}\).

Suppose the mechanism has approximation ratio \(\alpha\), and let
\(\beta=\frac{1}{\alpha}\). Therefore we have \(W_{ab}\ge \beta\operatorname{OPT}(a,b)\).
We now compute the optimal social utility for the four profiles used in
the proof:
\[
\begin{aligned}
    \operatorname{OPT}(1,5)
    &=\max\{1+5,16-1-5\}=10,
    &\operatorname{OPT}(1,6)
    =\max\{1+6,16-1-6\}=9,\\
    \operatorname{OPT}(2,7)
    &=\max\{2+7,16-2-7\}=9,
    &\operatorname{OPT}(3,7)
    =\max\{3+7,16-3-7\}=10.
\end{aligned}
\]
Thus,
\begin{align}
    W_{15}\ge 10\beta,\qquad
    W_{16}\ge 9\beta,\qquad
    W_{27}\ge 9\beta,\qquad
    W_{37}\ge 10\beta.
    \label{eq:manual-lower-welfare-constraints}
\end{align}
Multiplying these four inequalities by \(72,60,60,72\), respectively,
gives
\begin{align}
    72W_{15}+60W_{16}+60W_{27}+72W_{37}
    \ge 2520\beta.
    \label{eq:manual-lower-weighted-welfare}
\end{align}

We next add a nonnegative linear combination of strategyproofness
constraints. The resulting expression can be grouped by profile into functions of the facility location, each with a computable pointwise upper bound. This yields an upper bound independent of the mechanism's output distributions. Each term below is nonnegative by strategyproofness:
\begin{align*}
    \Sigma
    =&\ 15(D_{15}(5)-D_{16}(5))
      +42(D_{25}(2)-D_{15}(2)) +20(D_{25}(5)-D_{26}(5))\\
      &
      +15(D_{25}(5)-D_{27}(5)) +30(D_{26}(2)-D_{16}(2))
      +30(D_{26}(6)-D_{27}(6)) \\
      &+45(D_{35}(3)-D_{15}(3))
      +49(D_{35}(3)-D_{25}(3)) +49(D_{35}(5)-D_{36}(5))\\
      &
      +45(D_{35}(5)-D_{37}(5)) +15(D_{36}(3)-D_{16}(3))
      +20(D_{36}(3)-D_{26}(3))\\
      &+42(D_{36}(6)-D_{37}(6))
      +15(D_{37}(3)-D_{27}(3)) +141(D_{45}(4)-D_{35}(4)) \ge 0.
\end{align*}
For example, \(D_{15}(5)-D_{16}(5)\ge0\) says that an agent whose true
location is \(5\), while the other agent is at \(1\), cannot benefit by
misreporting \(6\). The other terms are interpreted similarly. Hence
\(\Sigma\ge0\).

Define $\Phi$ as follows; the inequality follows from \eqref{eq:manual-lower-weighted-welfare} and the nonnegativity of the strategyproofness constraints:
\begin{align}
    \Phi=72W_{15}+60W_{16}+60W_{27}+72W_{37}+\Sigma\ge 2520\beta. \label{eq:manual-lower-phi-lower}
\end{align}

It remains to upper bound \(\Phi\). Let \(d_t(y)=|t-y|\). Recall that
\[
    D_{ab}(t)=\mathbb{E}[|t-Y_{ab}|]
    =
    \mathbb{E}[d_t(Y_{ab})].
\]
Thus, once all terms involving the same profile \((a,b)\) are collected,
they can be written as the expectation of a single function of the random
outcome \(Y_{ab}\).

More precisely, for each profile \((a,b)\), we define an auxiliary
function \(\phi_{ab}:[0,8]\to\mathbb{R}\) as follows: collect all terms in
\(\Phi\) involving \(D_{ab}(\cdot)\), and replace every \(D_{ab}(t)\) by
\(d_t(y)\). Then the contribution of profile \((a,b)\) to \(\Phi\) is \(\mathbb{E}\left[\phi_{ab}(Y_{ab})\right]\).
Therefore
\[
    \Phi
    =
    \sum_{(a,b)}
    \mathbb{E}\left[\phi_{ab}(Y_{ab})\right],
\]
where profiles with zero total coefficient are omitted.

For example, consider profile \((1,5)\). From the weighted welfare term,
we get
\(72W_{15}
    =
    72D_{15}(1)+72D_{15}(5)\).
From the strategyproofness slacks in \(\Sigma\), the terms involving
\(Y_{15}\) are \(15D_{15}(5)-42D_{15}(2)-45D_{15}(3)\).
Hence the total contribution of profile \((1,5)\) is
\[
    72D_{15}(1)+87D_{15}(5)-42D_{15}(2)-45D_{15}(3)
    =
    \mathbb{E}\left[\phi_{15}(Y_{15})\right],
\]
where
\[
    \phi_{15}(y)
    =
    72d_1(y)+87d_5(y)-42d_2(y)-45d_3(y).
\]
The other functions \(\phi_{ab}\) are obtained in the same way. They are
listed in Table~\ref{tab:manual-lower-coefficients}. In the table, \(d_t\) abbreviates \(d_t(y)\).

\begin{center}
\captionof{table}{Coefficient functions and their pointwise maxima
in the social-utility lower-bound certificate.}
\label{tab:manual-lower-coefficients}
\renewcommand{\arraystretch}{1.15}
\begin{tabular}{clc}
\toprule
Profile & \(\phi_{ab}(y)\) & \(\max_{y\in[0,8]}\phi_{ab}(y)\)\\
\midrule
\((1,5)\)
&
\(72d_1+87d_5-42d_2-45d_3\)
&
\(288\)
\\

\((1,6)\)
&
\(60d_1+60d_6-15d_5-30d_2-15d_3\)
&
\(240\)
\\

\((2,5)\)
&
\(42d_2+35d_5-49d_3\)
&
\(112\)
\\

\((2,6)\)
&
\(30d_2+30d_6-20d_5-20d_3\)
&
\(80\)
\\

\((2,7)\)
&
\(60d_2+60d_7-15d_5-30d_6-15d_3\)
&
\(240\)
\\

\((3,5)\)
&
\(94d_3+94d_5-141d_4\)
&
\(188\)
\\

\((3,6)\)
&
\(35d_3+42d_6-49d_5\)
&
\(112\)
\\

\((3,7)\)
&
\(87d_3+72d_7-45d_5-42d_6\)
&
\(288\)
\\

\((4,5)\)
&
\(141d_4\)
&
\(564\)
\\
\bottomrule
\end{tabular}
\end{center}

Each function \(\phi_{ab}\) is piecewise linear, with breakpoints only at
integer points in \(\{0,1,\ldots,8\}\) because $d_t(y)=|t-y|$. Therefore its maximum over \([0,8]\) is attained at one of these integer
points. Table~\ref{tab:manual-lower-values} lists the values at all
integer points.

\begin{center}
\captionof{table}{Integer-point evaluations of the functions
\(\phi_{ab}\) used in Theorem~\ref{thm:su-lower}.}
\label{tab:manual-lower-values}
\label{app:manual-lower-table}
\small
\setlength{\tabcolsep}{4pt}
\renewcommand{\arraystretch}{1.1}
\begin{tabular}{crrrrrrrrrr}
\toprule
Profile
& \(y=0\) & \(y=1\) & \(y=2\) & \(y=3\) & \(y=4\)
& \(y=5\) & \(y=6\) & \(y=7\) & \(y=8\)
& \(\max_y \phi_{ab}(y)\)
\\
\midrule
\((1,5)\)
& 288 & 216 & 288 & 276 & 174 & 72 & 144 & 216 & 288
& 288
\\
\((1,6)\)
& 240 & 180 & 240 & 240 & 210 & 180 & 120 & 180 & 240
& 240
\\
\((2,5)\)
& 112 & 84 & 56 & 112 & 70 & 28 & 56 & 84 & 112
& 112
\\
\((2,6)\)
& 80 & 60 & 40 & 80 & 80 & 80 & 40 & 60 & 80
& 80
\\
\((2,7)\)
& 240 & 180 & 120 & 180 & 210 & 240 & 240 & 180 & 240
& 240
\\
\((3,5)\)
& 188 & 141 & 94 & 47 & 188 & 47 & 94 & 141 & 188
& 188
\\
\((3,6)\)
& 112 & 84 & 56 & 28 & 70 & 112 & 56 & 84 & 112
& 112
\\
\((3,7)\)
& 288 & 216 & 144 & 72 & 174 & 276 & 288 & 216 & 288
& 288
\\
\((4,5)\)
& 564 & 423 & 282 & 141 & 0 & 141 & 282 & 423 & 564
& 564
\\
\bottomrule
\end{tabular}

\end{center}

Summing the maxima in the last column gives
\begin{align}
    \Phi
    \le
    288+240+112+80+240+188+112+288+564
    =
    2112.
    \label{eq:manual-lower-phi-upper}
\end{align}

Combining \eqref{eq:manual-lower-phi-lower} and
\eqref{eq:manual-lower-phi-upper}, we get \(2520\beta\le 2112\). Hence
\[
    \beta\le \frac{2112}{2520}=\frac{88}{105} \Longleftrightarrow \alpha\ge \frac{105}{88}.
\]
This completes the proof.
\end{proof}

\section{Minimum Utility}

In this section, we analyze the minimum utility for both upper and lower bounds.

\subsection{Upper Bound}

Chan et al.~\cite{chan2025obnoxious} studied the mechanism that selects a
facility uniformly at random from the whole segment. They showed that this
mechanism gives an \(O(\sqrt n)\)-approximation for minimum utility. A more
careful calculation gives an approximation ratio of \(\sqrt{2n}+O(1)\) in their proof and analysis.
In this subsection, we improve the leading constant by adding a
small amount of probability mass to the two endpoints. The resulting
mechanism achieves an approximation ratio of \(\sqrt n+O(1)\).
For a parameter \(q\in[0,\frac12]\), define the following mechanism.

\begin{mechanism}\label{mec:mu-endpoint-uniform}
(Endpoint-Augmented Uniform Mechanism).
The mechanism outputs \(0\) with probability \(q\), outputs \(1\) with
probability \(q\), and with the remaining probability \(1-2q\), outputs a
point drawn uniformly at random from \([0,1]\).
\end{mechanism}

The mechanism is independent of the reported profile, and hence is
universally strategyproof. We set \(c=1-\frac{1}{\sqrt2}\) and \(q_n=\frac{c}{\sqrt n}\).
Since \(c<\frac12\), we have \(q_n\in[0,\frac12]\) for every \(n\ge1\).

\begin{theorem}\label{thm:mu-upper}
For every \(n\ge2\), the Endpoint-Augmented Uniform Mechanism with
\(q=q_n\) is universally strategyproof and achieves an approximation ratio
at most \(\sqrt n+O(1)\)
for minimum utility.
\end{theorem}

\begin{proof}
The mechanism is universally strategyproof because its distribution over
facility locations is independent of the reported profile. It remains to
prove the approximation ratio.

Fix a profile \(\mathbf{x}\), and write the agent locations in
nondecreasing order as \(x_1\le x_2\le\cdots\le x_n\).
Define the gaps
\[
    g_0=x_1,\qquad
    g_i=x_{i+1}-x_i\quad (1\le i\le n-1),\qquad
    g_n=1-x_n.
\]
Thus \(\sum_{i=0}^n g_i=1\). Let \(\Delta=\operatorname{OPT}_{\operatorname{MU}}(\mathbf{x})\)
be the optimal minimum utility. Since the best location is either an
endpoint of the segment or the midpoint of an empty interval between two
consecutive agents, we have
\[
    \Delta
    =
    \max\left\{
        g_0,\,
        g_n,\,
        \frac{g_1}{2},\ldots,\frac{g_{n-1}}{2}
    \right\}.
\]

Let \(U(\mathbf{x})\) be the expected minimum utility obtained by the
uniform distribution on \([0,1]\). On a boundary gap of length \(g_0\) or
\(g_n\), the nearest-agent distance forms a triangle of area
\(\frac12 g_0^2\) or \(\frac12 g_n^2\). On an internal gap of length
\(g_i\), the nearest-agent distance forms two symmetric triangles with
total area \(\frac14 g_i^2\). Therefore
\[
    U(\mathbf{x})
    =
    \frac{g_0^2}{2}
    +
    \frac{g_n^2}{2}
    +
    \sum_{i=1}^{n-1}\frac{g_i^2}{4}.
\]
Under the Endpoint-Augmented Uniform Mechanism, the expected minimum
utility is
\[
    \operatorname{ALG}(\mathbf{x})
    =
    q(g_0+g_n)+(1-2q)U(\mathbf{x}).
\]

We split the analysis according to whether the optimal empty interval is
an internal gap or a boundary gap.

First suppose that \(\Delta\) is attained by an internal gap. Then
\(g_i=2\Delta\) for some \(i\in\{1,\ldots,n-1\}\). The contribution of
this gap to \(U(\mathbf{x})\) is \(\Delta^2\). The remaining gaps have
total length \(1-2\Delta\). Ignoring the upper-bound constraints on these
remaining gaps can only decrease the minimum possible value of
\(U(\mathbf{x})\). Hence, by Cauchy's inequality,
\[
    U(\mathbf{x})
    \ge
    \Delta^2
    +
    \frac{(1-2\Delta)^2}{4(n-1)}.
\]
Indeed, among the remaining gaps, there are two boundary gaps with
coefficient \(\frac12\) and \(n-2\) internal gaps with coefficient
\(\frac14\), and
\[
    \frac{1}{1/2}+\frac{1}{1/2}+(n-2)\frac{1}{1/4}
    =
    4(n-1).
\]
Therefore
\[
    \operatorname{ALG}(\mathbf{x})
    \ge
    (1-2q)\left(
        \Delta^2+\frac{(1-2\Delta)^2}{4(n-1)}
    \right).
\]
For \(t>0\), define
\[
    F_I(t)
    =
    \frac{
        t
    }{
        t^2+\frac{(1-2t)^2}{4(n-1)}
    }.
\]
A direct calculation gives
\[
    F_I(t)
    \le
    \sqrt n+1.
\]
To see this, let
\[
    B_I(t)=t^2+\frac{(1-2t)^2}{4(n-1)}.
\]
Then
\[
    B_I(t)-tB_I'(t)
    =
    \frac{1}{4(n-1)}-\frac{n}{n-1}t^2,
\]
so the maximum is attained at \(t=\frac{1}{2\sqrt n}\), where
\(F_I(t)=\sqrt n+1\). Hence, in the internal-gap case,
\[
    \frac{\Delta}{\operatorname{ALG}(\mathbf{x})}
    \le
    \frac{\sqrt n+1}{1-2q}.
\]

Now suppose that \(\Delta\) is attained by a boundary gap. By symmetry,
assume without loss of generality that \(g_0=\Delta\). Then the endpoint
part of the mechanism contributes at least \(q\Delta\). Moreover, the
uniform part satisfies
\[
    U(\mathbf{x})
    \ge
    \frac{\Delta^2}{2}
    +
    \frac{(1-\Delta)^2}{4n-2}.
\]
Here the first term is the contribution of the boundary gap \(g_0\), and
the second term follows again from Cauchy's inequality applied to the
remaining gaps, whose total length is \(1-\Delta\). Among these remaining
gaps, there is one boundary gap with coefficient \(\frac12\) and
\(n-1\) internal gaps with coefficient \(\frac14\), so
\[
    \frac{1}{1/2}+(n-1)\frac{1}{1/4}=4n-2.
\]
Therefore
\[
    \operatorname{ALG}(\mathbf{x})
    \ge
    q\Delta
    +
    (1-2q)\left(
        \frac{\Delta^2}{2}
        +
        \frac{(1-\Delta)^2}{4n-2}
    \right).
\]
For \(t>0\), define
\[
    F_B(t)
    =
    \frac{
        t
    }{
        qt+(1-2q)\left(
            \frac{t^2}{2}+\frac{(1-t)^2}{4n-2}
        \right)
    }.
\]
Let \(r=1-2q\), and write the denominator as
\[
    qt+r\left(
        \frac{t^2}{2}+\frac{(1-t)^2}{4n-2}
    \right).
\]
The maximum of \(F_B(t)\) is attained when
\[
    t=\frac{1}{\sqrt{2n}}.
\]
Indeed, if \(B_B(t)\) denotes the denominator, then
\[
    B_B(t)-tB_B'(t)
    =
    r\left(
        \frac{1}{4n-2}
        -
        \frac{n}{2n-1}t^2
    \right),
\]
which vanishes exactly at \(t=\frac{1}{\sqrt{2n}}\). Substituting this
value gives
\[
    F_B(t)
    \le
    \left(
        q
        +
        (1-2q)\frac{\sqrt{2n}-1}{2n-1}
    \right)^{-1}.
\]

We now set \(q=q_n=\frac{c}{\sqrt n}\), where
\(c=1-\frac{1}{\sqrt2}\). In the internal-gap case, we get
\[
    \frac{\Delta}{\operatorname{ALG}(\mathbf{x})}
    \le
    \frac{\sqrt n+1}{1-\frac{2c}{\sqrt n}}
    =
    \sqrt n+O(1).
\]
In the boundary-gap case, we get
\[
    \frac{\Delta}{\operatorname{ALG}(\mathbf{x})}
    \le
    \left(
        \frac{c}{\sqrt n}
        +
        \left(1-\frac{2c}{\sqrt n}\right)
        \frac{\sqrt{2n}-1}{2n-1}
    \right)^{-1}.
\]
Since
\[
    \frac{\sqrt{2n}-1}{2n-1}
    =
    \frac{1}{\sqrt{2n}}+O\left(\frac1n\right),
\]
and since \(c+\frac{1}{\sqrt2}=1\), the denominator above is
\[
    \frac{1}{\sqrt n}+O\left(\frac1n\right).
\]
Thus the boundary-gap case also gives
\[
    \frac{\Delta}{\operatorname{ALG}(\mathbf{x})}
    \le
    \sqrt n+O(1).
\]
Combining the two cases proves the theorem.
\end{proof}

\subsection{Lower Bound}\label{subsec:mu-lower}

We next show that, for the minimum-utility objective and sufficiently large numbers of agents, no randomized strategyproof mechanism can achieve an approximation ratio strictly smaller than \(2\), which improves the previous value of $\frac{3}{2}$ of \cite{feigenbaum2020strategic}.

\begin{theorem}\label{thm:mu-lower}
    For every \(\rho<2\), there exists \(n_0\) such that, for every \(n\ge n_0\), no randomized strategyproof mechanism for \(n\) agents is a \(\rho\)-approximation for the minimum utility. Consequently, the asymptotic approximation ratio of randomized strategyproof mechanisms for minimum utility is at least \(2\).
\end{theorem}

We first record a simple monotonicity consequence of SP. For a randomized mechanism \(f\), let
\(\bar y(\mathbf{x})
    =
    \mathbb{E}_{Y\sim f(\mathbf{x})}[Y]\)
denote the expected facility location on profile \(\mathbf{x}\).

\begin{lemma}\label{lem:mu-mean-monotonicity}
    Let \(f\) be a randomized strategyproof mechanism. Suppose that a profile \(\mathbf{x}'\) is obtained from a profile \(\mathbf{x}\) by changing the reports of some agents located at \(0\), one by one, while keeping all other reports fixed. Then \(\bar y(\mathbf{x})\ge \bar y(\mathbf{x}')\). Similarly, if \(\mathbf{x}'\) is obtained from \(\mathbf{x}\) by changing the reports of some agents located at \(1\), one by one, then \(\bar y(\mathbf{x})\le \bar y(\mathbf{x}')\).
\end{lemma}

\begin{proof}
    It is enough to prove the statement for one changed report, since the general statement follows by iteration. Suppose first that \(\mathbf{x}'=(x_i',\mathbf{x}_{-i})\) and \(x_i=0\). By SP, agent \(i\) cannot increase her expected utility by reporting \(x_i'\) instead of \(0\). Since the facility always lies in \([0,1]\), this gives
    \[
        \mathbb{E}_{Y\sim f(\mathbf{x})}[Y]
        =
        \mathbb{E}_{Y\sim f(\mathbf{x})}[|0-Y|]
        \ge
        \mathbb{E}_{Y\sim f(\mathbf{x}')}[|0-Y|]
        =
        \mathbb{E}_{Y\sim f(\mathbf{x}')}[Y].
    \]
    Hence \(\bar y(\mathbf{x})\ge \bar y(\mathbf{x}')\). The case \(x_i=1\) is analogous. In that case,
    \[
        \mathbb{E}_{Y\sim f(\mathbf{x})}[1-Y]
        =
        \mathbb{E}_{Y\sim f(\mathbf{x})}[|1-Y|]
        \ge
        \mathbb{E}_{Y\sim f(\mathbf{x}')}[|1-Y|]
        =
        \mathbb{E}_{Y\sim f(\mathbf{x}')}[1-Y],
    \]
    and therefore \(\bar y(\mathbf{x})\le \bar y(\mathbf{x}')\).
\end{proof}

\begin{proof}[Proof of Theorem~\ref{thm:mu-lower}]
    It suffices to consider \(1\le\rho<2\), since approximation ratios are at least \(1\). Fix such a \(\rho\), and choose \(\varepsilon\in(0,1)\) sufficiently small so that
    \[
        \lambda
        :=
        \frac{1/\rho-\varepsilon}{1-\varepsilon}
        >
        \frac12 .
    \]
    This is possible because \(1/\rho>1/2\). Next choose \(\delta\in(0,1/2)\) sufficiently small so that
    \[
        \lambda(1-\delta)>\frac12 ,
    \]
    and set \(\eta=\varepsilon\delta\).

    We construct three profiles \(\mathbf{x}^0\), \(\mathbf{x}^R\), and \(\mathbf{x}^L\) using finite grids whose consecutive points are at distance at most \(\eta\). Set
    \[
        m=\left\lceil\frac{1-2\delta}{\eta}\right\rceil,
        \qquad k=\left\lceil\frac{\delta}{\eta}\right\rceil,
        \qquad n_0=m+1+2k.
    \]
    Fix any \(n\ge n_0\). In \(\mathbf{x}^0\), place one agent at each of the \(m+1\) equally spaced grid points of \([\delta,1-\delta]\), including both endpoints, and place \(k\) agents at each of \(0\) and \(1\). Place all \(n-n_0\) extra agents at \(\delta\), and keep these extra agents fixed in all three profiles.

    To obtain \(\mathbf{x}^R\), move \(k-1\) of the agents at \(0\), one to each point \(j\delta/k\), \(j=1,\ldots,k-1\). Keep every other report fixed, including one agent at \(0\). The occupied locations in \([0,1-\delta]\) include both endpoints and have consecutive gaps at most \(\eta\). Since \(\eta<\delta\) and agents remain at \(1\), the interval \([1-\delta,1]\) is the unique largest empty interval, and
    \[
        \operatorname{OPT}_{\operatorname{MU}}(\mathbf{x}^R)
        =\frac{\delta}{2}.
    \]
    Similarly, obtain \(\mathbf{x}^L\) from \(\mathbf{x}^0\) by moving \(k-1\) agents from \(1\), one to each point \(1-j\delta/k\), \(j=1,\ldots,k-1\), and keeping every other report fixed. Now the occupied locations in \([\delta,1]\) include both endpoints and have consecutive gaps at most \(\eta\), while \([0,\delta]\) is the unique largest empty interval. Hence
    \[
        \operatorname{OPT}_{\operatorname{MU}}(\mathbf{x}^L)
        =\frac{\delta}{2}.
    \]

    Suppose, for contradiction, that there exists a randomized strategyproof mechanism \(f\) that is a \(\rho\)-approximation for the minimum-utility objective. We first consider profile \(\mathbf{x}^R\). Let \(Y^R\sim f(\mathbf{x}^R)\), and let
    \[
        q_R
        =
        \Pr[Y^R\in[1-\delta,1]].
    \]
    On \([1-\delta,1]\), the distance from any point to the nearest agent in \(\mathbf{x}^R\) is at most \(\delta/2\). On \([0,1-\delta]\), since both endpoints are occupied and consecutive occupied locations are at distance at most \(\eta\), the distance from any point to the nearest agent is at most \(\eta/2\). Therefore,
    \[
        \operatorname{MU}(\mathbf{x}^R,f(\mathbf{x}^R))
        \le
        q_R\cdot\frac{\delta}{2}
        +
        (1-q_R)\cdot\frac{\eta}{2}.
    \]
    On the other hand, since \(f\) is assumed to be a \(\rho\)-approximation,
    \[
        \operatorname{MU}(\mathbf{x}^R,f(\mathbf{x}^R))
        \ge
        \frac{\operatorname{OPT}_{\operatorname{MU}}(\mathbf{x}^R)}{\rho}
        =
        \frac{\delta}{2\rho}.
    \]
    Combining the two inequalities gives
    \[
        q_R\delta+(1-q_R)\eta
        \ge
        \frac{\delta}{\rho},
    \]
    and hence
    \[
        q_R
        \ge
        \frac{\delta/\rho-\eta}{\delta-\eta}
        =
        \frac{1/\rho-\varepsilon}{1-\varepsilon}
        =
        \lambda .
    \]
    Since every point in \([1-\delta,1]\) is at least \(1-\delta\), we obtain
    \[
        \bar y(\mathbf{x}^R)
        =
        \mathbb{E}[Y^R]
        \ge
        q_R(1-\delta)
        \ge
        \lambda(1-\delta)
        >
        \frac12 .
    \]
    Moreover, \(\mathbf{x}^R\) is obtained from \(\mathbf{x}^0\) by changing only reports of agents located at \(0\). By Lemma~\ref{lem:mu-mean-monotonicity},
    \[
        \bar y(\mathbf{x}^0)
        \ge
        \bar y(\mathbf{x}^R)
        >
        \frac12 .
    \]

    We now consider profile \(\mathbf{x}^L\). Let \(Y^L\sim f(\mathbf{x}^L)\), and let
    \[
        q_L
        =
        \Pr[Y^L\in[0,\delta]].
    \]
    By the same argument as above, since \([0,\delta]\) is the unique largest empty interval of \(\mathbf{x}^L\) and the rest of the segment has occupied endpoints and consecutive gaps at most \(\eta\), the \(\rho\)-approximation ratio implies
    \[
        q_L
        \ge
        \frac{\delta/\rho-\eta}{\delta-\eta}
        =
        \lambda .
    \]
    Since every point in \([0,\delta]\) is at most \(\delta\), while every point in \([0,1]\) is at most \(1\), we have
    \[
        \bar y(\mathbf{x}^L)
        =
        \mathbb{E}[Y^L]
        \le
        q_L\delta+(1-q_L)
        =
        1-q_L(1-\delta)
        \le
        1-\lambda(1-\delta)
        <
        \frac12 .
    \]
    Moreover, \(\mathbf{x}^L\) is obtained from \(\mathbf{x}^0\) by changing only reports of agents located at \(1\). By Lemma~\ref{lem:mu-mean-monotonicity},
    \[
        \bar y(\mathbf{x}^0)
        \le
        \bar y(\mathbf{x}^L)
        <
        \frac12 .
    \]

    We have derived both \(\bar y(\mathbf{x}^0)>1/2\) and \(\bar y(\mathbf{x}^0)<1/2\), a contradiction. Therefore no randomized strategyproof mechanism can be a \(\rho\)-approximation for any \(\rho<2\), provided that \(n\) is sufficiently large. Since \(\rho<2\) was arbitrary, the asymptotic lower bound is \(2\).
\end{proof}

\section{Conclusion}\label{sec:conclusion}

We studied randomized strategyproof mechanisms for obnoxious facility
location on a line segment under the social utility and minimum utility
objectives. For social utility, we improved both the best known upper
bound and lower bound for randomized strategyproof mechanisms. For minimum
utility, we improved the leading constant in the known \(O(\sqrt n)\)
upper bound and proved an asymptotic lower bound of \(2\). Several gaps
remain open. For social utility, it is open whether a sharper kernel
certificate can improve the approximation guarantee for the same
random-threshold dictator mechanism. For minimum utility, the true asymptotic order remains far
from settled: while we suspect that an \(\Omega(\sqrt n)\) lower bound may
hold, our current techniques only prove a constant lower bound. Closing
this gap, even by improving the lower bound beyond \(2\), seems to be a
promising direction for future work.

\paragraph*{AI Assistance Disclosure.}
The authors used ChatGPT to assist with proofreading, language polishing,
and checking selected numerical details. All mathematical results and
claims were independently verified by the authors, who take full
responsibility for the content of the paper.


\bibliographystyle{plain}
\bibliography{myreferences}

@article{alon2010strategyproof,
  title={Strategyproof approximation of the minimax on networks},
  author={Alon, Noga and Feldman, Michal and Procaccia, Ariel D and Tennenholtz, Moshe},
  journal={Mathematics of Operations Research},
  volume={35},
  number={3},
  pages={513--526},
  year={2010},
  publisher={INFORMS}
}

@inproceedings{alex2024,
  author = {Lam, Alexander and Aziz, Haris and Li, Bo and Ramezani, Fahimeh and Walsh, Toby},
  title = {Proportional fairness in obnoxious facility location},
  year = {2024},
  booktitle = {Proceedings of the 23rd International Conference on Autonomous Agents and Multiagent Systems},
  pages = {1075--1083}
}

@article{aziz2020capacity,
  title={The capacity constrained facility location problem},
  author={Aziz, Haris and Chan, Hau and Lee, Barton E and Parkes, David C},
  journal={Games and Economic Behavior},
  volume={124},
  pages={478--490},
  year={2020},
  publisher={Elsevier}
}

@article{cheng2013strategy,
  title={Strategy-proof approximation mechanisms for an obnoxious facility game on networks},
  author={Cheng, Yukun and Yu, Wei and Zhang, Guochuan},
  journal={Theoretical Computer Science},
  volume={497},
  pages={154--163},
  year={2013},
  publisher={Elsevier}
}

@article{d1979hotelling,
  title={On {Hotelling}'s ``stability in competition''},
  author={d'Aspremont, Claude and Gabszewicz, J Jaskold and Thisse, J-F},
  journal={Econometrica: Journal of the Econometric Society},
  pages={1145--1150},
  year={1979},
  publisher={JSTOR}
}

@article{feigenbaum2017approximately,
  title={Approximately optimal mechanisms for strategyproof facility location: minimizing {$L_p$} norm of costs},
  author={Feigenbaum, Itai and Sethuraman, Jay and Ye, Chun},
  journal={Mathematics of Operations Research},
  volume={42},
  number={2},
  pages={434--447},
  year={2017},
  publisher={INFORMS}
}

@article{GoelH23,
  author       = {Sumit Goel and
                  Wade Hann{-}Caruthers},
  title        = {Optimality of the coordinate-wise median mechanism for strategyproof facility location in two dimensions},
  journal      = {Social Choice and Welfare},
  volume       = {61},
  number       = {1},
  pages        = {11--34},
  year         = {2023}
}

@inproceedings{gravin2025approximation,
  title={Approximation guarantees of median mechanism in $\mathbb{R}^d$},
  author={Gravin, Nikolai and Jia, Jianhao},
  booktitle={Proceedings of the 57th Annual ACM Symposium on Theory of Computing (STOC)},
  pages={495--506},
  year={2025}
}

@inproceedings{Ibara2012characterize,
  title={Characterizing mechanisms in obnoxious facility game},
  author={Ibara, Ken and Nagamochi, Hiroshi},
  booktitle={International Conference on Combinatorial Optimization and Applications},
  pages={301--311},
  year={2012},
  organization={Springer}
}

@article{DBLP:journals/corr/abs-2212-09521,
  author       = {Gabriel Istrate and
                  Cosmin Bonchis},
  title        = {Mechanism design with predictions for obnoxious facility location},
  journal      = {CoRR},
  volume       = {abs/2212.09521},
  year         = {2022},
  eprinttype    = {arXiv}
}

@inproceedings{li2024strategyproof,
  title={Strategyproof mechanisms for group-fair obnoxious facility location problems},
  author={Li, Jiaqian and Li, Minming and Chan, Hau},
  booktitle={Proceedings of the AAAI Conference on Artificial Intelligence},
  volume={38},
  pages={9832--9839},
  year={2024}
}

@inproceedings{lin2020nearly,
  title={Nearly complete characterization of 2-agent deterministic strategyproof mechanisms for single facility location in $L_p$ space},
  author={Lin, Jianan},
  booktitle={International Conference on Combinatorial Optimization and Applications},
  pages={411--425},
  year={2020},
  organization={Springer}
}

@inproceedings{lu2009tighter,
  title={Tighter bounds for facility games},
  author={Lu, Pinyan and Wang, Yajun and Zhou, Yuan},
  booktitle={Proceedings of the 5th International Workshop on Internet and Network Economics (WINE)},
  pages={137--148},
  year={2009},
  organization={Springer}
}

@inproceedings{lu10mechanism,
  author = {Lu, Pinyan and Sun, Xiaorui and Wang, Yajun and Zhu, Zeyuan Allen},
  title = {Asymptotically optimal strategy-proof mechanisms for two-facility games},
  year = {2010},
  booktitle = {Proceedings of the 11th ACM Conference on Electronic Commerce (EC)},
  pages = {315--324}
}

@inproceedings{DBLP:conf/sagt/Meir19,
  author       = {Reshef Meir},
  title        = {Strategyproof facility location for three agents on a circle},
  booktitle    = {Algorithmic Game Theory - 12th International Symposium (SAGT)
                 },
  series       = {Lecture Notes in Computer Science},
  volume       = {11801},
  pages        = {18--33},
  publisher    = {Springer},
  year         = {2019}
}

@article{moulin1980strategy,
  title={On strategy-proofness and single peakedness},
  author={Moulin, Herv{\'e}},
  journal={Public Choice},
  volume={35},
  number={4},
  pages={437--455},
  year={1980},
  publisher={Springer}
}

@article{oomine2016characterizing,
  title={Characterizing output locations of {GSP} mechanisms to obnoxious facility game in trees},
  author={Oomine, Morito and Nagamochi, Hiroshi},
  journal={IEICE TRANSACTIONS on Information and Systems},
  volume={E99-D},
  number={3},
  pages={615--623},
  year={2016},
  publisher={The Institute of Electronics, Information and Communication Engineers}
}

@article{procaccia2013approximate,
  title={Approximate mechanism design without money},
  author={Procaccia, Ariel D and Tennenholtz, Moshe},
  journal={ACM Transactions on Economics and Computation (TEAC)},
  volume={1},
  number={4},
  pages={1--26},
  year={2013},
  publisher={ACM New York, NY, USA}
}

@inproceedings{tang2020characterization,
  title={Characterization of group-strategyproof mechanisms for facility location in strictly convex space},
  author={Tang, Pingzhong and Yu, Dingli and Zhao, Shengyu},
  booktitle={Proceedings of the 21st ACM Conference on Economics and Computation (EC)},
  pages={133--157},
  year={2020}
}

@inproceedings{walsh24utility,
  author    = {Toby Walsh},
  title     = {Approximate mechanism design for facility location with multiple objectives},
  booktitle = {{ECAI} 2024 - 27th European Conference on Artificial Intelligence},
  pages     = {3533-3540},
  year      = {2024}
}

@inproceedings{ye2015strategy,
  title={Strategy-proof mechanism for obnoxious facility location on a line},
  author={Ye, Deshi and Mei, Lili and Zhang, Yong},
  booktitle={International Computing and Combinatorics Conference},
  pages={45--56},
  year={2015},
  organization={Springer}
}

@inproceedings{chan2025obnoxious,
  title={Obnoxious Facility Location Problems: Strategyproof Mechanisms Optimizing $L_p$-Aggregated Utilities and Costs},
  author={Chan, Hau and Lin, Jianan and Wang, Chenhao},
  booktitle={Proceedings of the 2026 International Conference on Autonomous Agents and Multiagent Systems (AAMAS)},
  pages={496--504},
  year={2026}
}

@article{barak2026facility,
  title={Facility Location Mechanism Design--Breaking The Deterministic Barrier},
  author={Barak, Zohar},
  journal={arXiv preprint arXiv:2605.24750},
  year={2026},
  note={To appear in EC 26}
}

@article{feigenbaum2020strategic,
  title={Strategic facility location problems with linear single-dipped and single-peaked preferences},
  author={Feigenbaum, Itai and Li, Minming and Sethuraman, Jay and Wang, Fangzhou and Zou, Shaokun},
  journal={Autonomous Agents and Multi-Agent Systems},
  volume={34},
  number={2},
  pages={49},
  year={2020},
  publisher={Springer}
}

@article{chan2026strategyproof,
  title={Strategyproof Mechanisms for Euclidean Facility Location Problems under $L_p$-norm Social Cost},
  author={Chan, Hau and Lin, Jianan and Wang, Chenhao},
  journal={arXiv preprint arXiv:2606.08621},
  year={2026}
}

@article{chan2026randomized,
  title={Randomized Strategyproof Facility Location: Two Facilities and Beyond},
  author={Chan, Hau and Lin, Jianan and Wang, Chenhao},
  journal={arXiv preprint arXiv:2608.22484},
  year={2026}
}

@inproceedings{rogowski2025improved,
  title={Improved approximation ratio for strategyproof facility location on a cycle},
  author={Rogowski, Krzysztof and Dziubi{\'n}ski, Marcin},
  booktitle={Proceedings of the Thirty-Fourth International Joint Conference on Artificial Intelligence},
  pages={4032--4039},
  year={2025}
}

\appendix

\section{Quadratic Certificate for the Social Utility Upper Bound}
\label{app:upper-kernel}

In this appendix we provide the details for the quadratic certificate used
in Lemma~\ref{lem:quadratic-kernel-lower-bound}. Recall that
\[
    D(x,z)=K_{93/100}^s(x,z)-P(x,z),
\]
where
\[
    P(x,z)=
    \frac{7641}{10000}
    -\frac{2741}{4000}(x+z)
    +\frac{6869}{5000}xz.
\]
By symmetry, it suffices to consider \(0\le x\le z\le1\). We use the
regions \(U_{AA},U_{AB},U_{CA},U_{CB}\) and \(V_L,V_X,V_R\) defined in
\eqref{eq:uniform-regions} and \eqref{eq:atom-regions}. On each nonempty
intersection, \(D(x,z)\) is a quadratic polynomial.

For each nonempty intersection \(S\), let \(Q_S\) denote the polynomial
that agrees with \(10000D\) on \(S\), extended to its closure.
Table~\ref{tab:su-upper-certificate} gives \(Q_S\), a minimizer over \(\overline S\),
and \(m_S:=\frac{1}{10000}\min_{\overline S}Q_S\). This last value is a lower bound
for \(D\) on \(S\); the extension need not agree with \(10000D\)
at boundary points outside \(S\).

\begin{longtable}{c|p{0.46\textwidth}|c|c}
\caption{Quadratic certificate for the social-utility upper bound:
polynomial extensions, minimizers over the region closures, and lower bounds.}
\label{tab:su-upper-certificate}\\
\toprule
\(S\) & \(Q_S(x,z)\) & Minimizer in \(\overline S\) & \(m_S\)\\
\midrule
\endfirsthead
\toprule
\(S\) & \(Q_S(x,z)\) & Minimizer in \(\overline S\) & \(m_S\)\\
\midrule
\endhead

\(U_{AA}\cap V_L\)
&
\(-4650x^2+4862xz+\frac{3705}{2}x-4650z^2+\frac{3705}{2}z+34\)
&
\((0,0)\)
&
\(\frac{17}{5000}\)
\\[4pt]

\(U_{AA}\cap V_X\)
&
\(-4650x^2+4862xz+\frac{5105}{2}x-4650z^2+\frac{3705}{2}z-316\)
&
\(\left(\frac13,\frac23\right)\)
&
\(\frac{961}{36000}\)
\\[4pt]

\(U_{AA}\cap V_R\)
&
\(-4650x^2+4862xz+\frac{5105}{2}x-4650z^2+\frac{5105}{2}z-666\)
&
\((1,1)\)
&
\(\frac{1}{10000}\)
\\[4pt]

\(U_{AB}\cap V_L\)
&
\(4650x^2-4438xz+\frac{3705}{2}x-2325z^2+\frac{3705}{2}z+34\)
&
\((0,0)\)
&
\(\frac{17}{5000}\)
\\[4pt]

\(U_{AB}\cap V_X\)
&
\(4650x^2-4438xz+\frac{5105}{2}x-2325z^2+\frac{3705}{2}z-316\)
&
\(\left(0,\frac12\right)\)
&
\(\frac{29}{10000}\)
\\[4pt]

\(U_{CA}\cap V_X\)
&
\(-2325x^2-4438xz+\frac{14405}{2}x+4650z^2-\frac{14895}{2}z+2009\)
&
\(\left(\frac12,1\right)\)
&
\(\frac{1}{800}\)
\\[4pt]

\(U_{CA}\cap V_R\)
&
\(-2325x^2-4438xz+\frac{14405}{2}x+4650z^2-\frac{13495}{2}z+1659\)
&
\((1,1)\)
&
\(\frac{1}{10000}\)
\\[4pt]

\(U_{CB}\cap V_L\)
&
\(6975x^2-13738xz+\frac{13005}{2}x+6975z^2-\frac{14895}{2}z+2359\)
&
\(\left(0,\frac12\right)\)
&
\(\frac{379}{10000}\)
\\[4pt]

\(U_{CB}\cap V_X\)
&
\(6975x^2-13738xz+\frac{14405}{2}x+6975z^2-\frac{14895}{2}z+2009\)
&
\(\left(\frac{57465}{183433},\frac{618085}{733732}\right)\)
&
\(\frac{1027577}{29349280000}\)
\\[4pt]

\bottomrule
\end{longtable}

All minimum values in the last column are nonnegative. Hence
\(D(x,z)\ge0\) on each region, which proves the pointwise lower bound
\eqref{eq:quadratic-kernel-lower-bound}.

There are \(4\cdot3=12\) possible intersections between the uniform
regions \(U_{AA},U_{AB},U_{CA},U_{CB}\) and the atom regions
\(V_L,V_X,V_R\). Three of them are empty. First, \(U_{AB}\cap V_R\) is
empty: if \((x,z)\in U_{AB}\cap V_R\), then \(x\ge\frac12\) and
\(z\ge2x\), so \(z\ge1\). Since \(z\le1\), this forces
\(x=\frac12\) and \(z=1\). But \(U_{AB}\) also requires
\(z\le\frac{1+x}{2}=\frac34\), a contradiction. Second,
\(U_{CA}\cap V_L\) is empty: if \((x,z)\in U_{CA}\cap V_L\), then
\(z\le\frac12\) and \(z\ge\frac{1+x}{2}\). Hence
\(1+x\le2z\le1\), so \(x=0\) and \(z=\frac12\). But \(U_{CA}\) also
requires \(z\le2x\), which would imply \(\frac12\le0\), again a
contradiction. Third, \(U_{CB}\cap V_R\) is empty because
\(x>\frac12\) and \(z\ge2x\) would imply \(z>1\).
In particular, \((\frac12,1)\) belongs to \(V_X\), not \(V_R\).

We first explain how the polynomial \(10000D(x,z)\) in each row of the
certificate table is obtained. This is the only place where the parameter
\(\lambda=\frac{93}{100}\) and the constants in \(P(x,z)\) enter the
calculation.

Recall that
\[
    D(x,z)=K_{93/100}^s(x,z)-P(x,z),
\]
where
\[
    K_{93/100}^s(x,z)
    =
    \frac{93}{100}K^s(x,z)
    +
    \frac{7}{100}H^s(x,z),
\]
and
\[
    P(x,z)
    =
    \frac{7641}{10000}
    -
    \frac{2741}{4000}(x+z)
    +
    \frac{6869}{5000}xz.
\]
Multiplying by \(10000\), we get
\begin{align}
    10000D(x,z)
    &=
    9300K^s(x,z)
    +
    700H^s(x,z)
    -
    10000P(x,z) \notag\\
    &=
    9300K^s(x,z)
    +
    700H^s(x,z)
    -
    7641
    +
    \frac{13705}{2}(x+z)
    -
    13738xz.
    \label{eq:Q-common-form}
\end{align}
Thus, on each intersection \(U\cap V\), the polynomial in the table is
obtained by substituting the corresponding formula for \(K^s\) on \(U\)
and the corresponding formula for \(H^s\) on \(V\) into
\eqref{eq:Q-common-form}.

We now list these formulas explicitly. On the triangle
\(0\le x\le z\le1\), the symmetrized uniform-threshold kernel \(K^s\) has
four possible expressions, corresponding to the four uniform regions:
\[
\begin{array}{c|c}
\text{uniform region} & K^s(x,z)\\
\hline
U_{AA}
&
\displaystyle
2xz-\frac{x+z}{2}-\frac{x^2+z^2}{2}+\frac34
\\[6pt]
U_{AB}
&
\displaystyle
xz+\frac{x^2}{2}-\frac{z^2}{4}
-\frac{x+z}{2}+\frac34
\\[6pt]
U_{CA}
&
\displaystyle
xz-\frac{x^2}{4}+\frac{z^2}{2}
-\frac{3z}{2}+1
\\[6pt]
U_{CB}
&
\displaystyle
\frac34x^2+\frac34z^2-\frac32z+1
\end{array}
\]
Similarly, the symmetrized atom kernel \(H^s\) has three possible
expressions:
\[
\begin{array}{c|c}
\text{atom region} & H^s(x,z)\\
\hline
V_L
&
\displaystyle
1-\frac{x+z}{2}
\\[6pt]
V_X
&
\displaystyle
\frac{x+1-z}{2}
\\[6pt]
V_R
&
\displaystyle
\frac{x+z}{2}
\end{array}
\]

We now derive the polynomial \(Q_{U,V}(x,z)=10000D(x,z)\) for every
nonempty intersection \(U\cap V\).

\paragraph*{Case 1: \(U_{AA}\cap V_L\).}
Substituting
\[
    K^s(x,z)
    =
    2xz-\frac{x+z}{2}-\frac{x^2+z^2}{2}+\frac34
\]
and
\[
    H^s(x,z)
    =
    1-\frac{x+z}{2}
\]
into \eqref{eq:Q-common-form}, we get
\begin{align*}
    Q_{AA,L}(x,z)
    &=
    9300\left(
        2xz-\frac{x+z}{2}-\frac{x^2+z^2}{2}+\frac34
    \right)
    +
    700\left(1-\frac{x+z}{2}\right)  \\
    &\quad
    -
    7641
    +
    \frac{13705}{2}(x+z)
    -
    13738xz  \\
    &=
    -4650x^2
    +
    4862xz
    +
    \frac{3705}{2}x
    -
    4650z^2
    +
    \frac{3705}{2}z
    +
    34.
\end{align*}
Here, for example, the \(xz\)-coefficient is
\(9300\cdot2-13738=4862\), the \(x^2\)- and \(z^2\)-coefficients are both
\(9300\cdot(-\frac12)=-4650\), and the constant term is
\(9300\cdot\frac34+700-7641=34\).

\paragraph*{Case 2: \(U_{AA}\cap V_X\).}
Here \(K^s\) is the same as above, while
\(H^s(x,z)=\frac{x+1-z}{2}\). Hence
\begin{align*}
    Q_{AA,X}(x,z)
    &=
    9300\left(
        2xz-\frac{x+z}{2}-\frac{x^2+z^2}{2}+\frac34
    \right)
    +
    700\left(\frac{x+1-z}{2}\right)  \\
    &\quad
    -
    7641
    +
    \frac{13705}{2}(x+z)
    -
    13738xz  \\
    &=
    -4650x^2
    +
    4862xz
    +
    \frac{5105}{2}x
    -
    4650z^2
    +
    \frac{3705}{2}z
    -
    316.
\end{align*}
For instance, the coefficient of \(x\) is
\(-9300\cdot\frac12+700\cdot\frac12+\frac{13705}{2}
=\frac{5105}{2}\), while the coefficient of \(z\) is
\(-9300\cdot\frac12-700\cdot\frac12+\frac{13705}{2}
=\frac{3705}{2}\).

\paragraph*{Case 3: \(U_{AA}\cap V_R\).}
Here \(K^s\) is again the \(U_{AA}\) expression, and
\(H^s(x,z)=\frac{x+z}{2}\). Therefore
\begin{align*}
    Q_{AA,R}(x,z)
    &=
    9300\left(
        2xz-\frac{x+z}{2}-\frac{x^2+z^2}{2}+\frac34
    \right)
    +
    700\left(\frac{x+z}{2}\right)  \\
    &\quad
    -
    7641
    +
    \frac{13705}{2}(x+z)
    -
    13738xz  \\
    &=
    -4650x^2
    +
    4862xz
    +
    \frac{5105}{2}x
    -
    4650z^2
    +
    \frac{5105}{2}z
    -
    666.
\end{align*}

\paragraph*{Case 4: \(U_{AB}\cap V_L\).}
On \(U_{AB}\), we have
\[
    K^s(x,z)
    =
    xz+\frac{x^2}{2}-\frac{z^2}{4}
    -\frac{x+z}{2}+\frac34.
\]
Together with \(H^s(x,z)=1-\frac{x+z}{2}\), this gives
\begin{align*}
    Q_{AB,L}(x,z)
    &=
    9300\left(
        xz+\frac{x^2}{2}-\frac{z^2}{4}
        -\frac{x+z}{2}+\frac34
    \right)
    +
    700\left(1-\frac{x+z}{2}\right)  \\
    &\quad
    -
    7641
    +
    \frac{13705}{2}(x+z)
    -
    13738xz  \\
    &=
    4650x^2
    -
    4438xz
    +
    \frac{3705}{2}x
    -
    2325z^2
    +
    \frac{3705}{2}z
    +
    34.
\end{align*}
Here the \(xz\)-coefficient is \(9300-13738=-4438\), the \(x^2\)-coefficient
is \(9300\cdot\frac12=4650\), and the \(z^2\)-coefficient is
\(9300\cdot(-\frac14)=-2325\).

\paragraph*{Case 5: \(U_{AB}\cap V_X\).}
Using the same \(U_{AB}\) expression for \(K^s\), and using
\(H^s(x,z)=\frac{x+1-z}{2}\), we get
\begin{align*}
    Q_{AB,X}(x,z)
    &=
    9300\left(
        xz+\frac{x^2}{2}-\frac{z^2}{4}
        -\frac{x+z}{2}+\frac34
    \right)
    +
    700\left(\frac{x+1-z}{2}\right)  \\
    &\quad
    -
    7641
    +
    \frac{13705}{2}(x+z)
    -
    13738xz  \\
    &=
    4650x^2
    -
    4438xz
    +
    \frac{5105}{2}x
    -
    2325z^2
    +
    \frac{3705}{2}z
    -
    316.
\end{align*}

\paragraph*{Case 6: \(U_{CA}\cap V_X\).}
On \(U_{CA}\), we have
\[
    K^s(x,z)
    =
    xz-\frac{x^2}{4}+\frac{z^2}{2}
    -\frac{3z}{2}+1.
\]
Together with \(H^s(x,z)=\frac{x+1-z}{2}\), this gives
\begin{align*}
    Q_{CA,X}(x,z)
    &=
    9300\left(
        xz-\frac{x^2}{4}+\frac{z^2}{2}
        -\frac{3z}{2}+1
    \right)
    +
    700\left(\frac{x+1-z}{2}\right)  \\
    &\quad
    -
    7641
    +
    \frac{13705}{2}(x+z)
    -
    13738xz  \\
    &=
    -2325x^2
    -
    4438xz
    +
    \frac{14405}{2}x
    +
    4650z^2
    -
    \frac{14895}{2}z
    +
    2009.
\end{align*}
For the linear terms, note that the \(x\)-coefficient is
\(700\cdot\frac12+\frac{13705}{2}=\frac{14405}{2}\), while the
\(z\)-coefficient is
\(9300\cdot(-\frac32)-700\cdot\frac12+\frac{13705}{2}
=-\frac{14895}{2}\).

\paragraph*{Case 7: \(U_{CA}\cap V_R\).}
Here \(K^s\) is again the \(U_{CA}\) expression, and
\(H^s(x,z)=\frac{x+z}{2}\). Thus
\begin{align*}
    Q_{CA,R}(x,z)
    &=
    9300\left(
        xz-\frac{x^2}{4}+\frac{z^2}{2}
        -\frac{3z}{2}+1
    \right)
    +
    700\left(\frac{x+z}{2}\right)  \\
    &\quad
    -
    7641
    +
    \frac{13705}{2}(x+z)
    -
    13738xz  \\
    &=
    -2325x^2
    -
    4438xz
    +
    \frac{14405}{2}x
    +
    4650z^2
    -
    \frac{13495}{2}z
    +
    1659.
\end{align*}

\paragraph*{Case 8: \(U_{CB}\cap V_L\).}
On \(U_{CB}\), we have
\[
    K^s(x,z)
    =
    \frac34x^2+\frac34z^2-\frac32z+1.
\]
Together with \(H^s(x,z)=1-\frac{x+z}{2}\), this gives
\begin{align*}
    Q_{CB,L}(x,z)
    &=
    9300\left(
        \frac34x^2+\frac34z^2-\frac32z+1
    \right)
    +
    700\left(1-\frac{x+z}{2}\right)  \\
    &\quad
    -
    7641
    +
    \frac{13705}{2}(x+z)
    -
    13738xz  \\
    &=
    6975x^2
    -
    13738xz
    +
    \frac{13005}{2}x
    +
    6975z^2
    -
    \frac{14895}{2}z
    +
    2359.
\end{align*}
Here \(K^s\) has no \(xz\)-term on \(U_{CB}\), so the entire
\(xz\)-coefficient \(-13738\) comes from \(-10000P(x,z)\).

\paragraph*{Case 9: \(U_{CB}\cap V_X\).}
Using the \(U_{CB}\) expression for \(K^s\), and using
\(H^s(x,z)=\frac{x+1-z}{2}\), we get
\begin{align*}
    Q_{CB,X}(x,z)
    &=
    9300\left(
        \frac34x^2+\frac34z^2-\frac32z+1
    \right)
    +
    700\left(\frac{x+1-z}{2}\right)  \\
    &\quad
    -
    7641
    +
    \frac{13705}{2}(x+z)
    -
    13738xz  \\
    &=
    6975x^2
    -
    13738xz
    +
    \frac{14405}{2}x
    +
    6975z^2
    -
    \frac{14895}{2}z
    +
    2009.
\end{align*}

The three intersections \(U_{AB}\cap V_R\), \(U_{CA}\cap V_L\), and \(U_{CB}\cap V_R\) are empty,
as verified above, so no polynomial is
needed for them. The expressions derived above are exactly the
polynomials \(10000D(x,z)\) appearing in the certificate table.

We now verify the polynomial minimum in each row of the table.
In each case below, \(S\) denotes the nonempty intersection named in the
heading, \(R=\overline S\), and \(Q_R\) denotes the continuous polynomial
extension of \(10000D|_S\). We compute
\(m_S=\frac{1}{10000}\min_R Q_R\), so that \(D(x,z)\ge m_S\)
for all \((x,z)\in S\). We do not identify this extension with \(D\)
at points of \(R\setminus S\). Each actual boundary point is covered
by its own atom region in \eqref{eq:atom-regions}.

For two points
\(p=(p_x,p_z)\) and \(q=(q_x,q_z)\), write
\[
    e_{p,q}(t)=(1-t)p+tq,
    \qquad 0\le t\le1.
\]
If an edge of a region is the segment from \(p\) to \(q\), then the
restriction of \(Q_R\) to this edge is the univariate quadratic
\[
    q_{p,q}(t)
    =
    Q_R(e_{p,q}(t)).
\]
The minimum of \(q_{p,q}\) over \(0\le t\le1\) is obtained by checking
\(t=0\), \(t=1\), and, if it lies in \([0,1]\), the solution of
\(q_{p,q}'(t)=0\). Since each \(Q_R\) is a quadratic polynomial and each
two-dimensional region is a polygon, the minimum over the region is
attained either at an interior stationary point or on one of the boundary
segments. Thus the following checks are exhaustive.

\paragraph*{Case 1: \(U_{AA}\cap V_L\).}
The region is
\[
    R=\left\{(x,z):0\le x\le z\le\frac12,\ z\le2x\right\},
\]
with vertices
\[
    p_0=(0,0),\qquad
    p_1=\left(\frac12,\frac12\right),\qquad
    p_2=\left(\frac14,\frac12\right).
\]
On this region,
\[
    Q_R(x,z)
    =
    -4650x^2+4862xz+\frac{3705}{2}x
    -4650z^2+\frac{3705}{2}z+34.
\]
The stationary equations are
\[
    -9300x+4862z+\frac{3705}{2}=0,
    \qquad
    4862x-9300z+\frac{3705}{2}=0.
\]
Solving this linear system gives
\[
    (x,z)=\left(\frac{3705}{8876},\frac{3705}{8876}\right).
\]
This point is feasible, but its value is
\[
    Q_R\left(\frac{3705}{8876},\frac{3705}{8876}\right)
    =
    \frac{14330593}{17752}
    >
    34.
\]
It remains to check the three edges:
\[
\renewcommand{\arraystretch}{1.25}
\begin{array}{c|c|c|c|c}
\text{edge} & (x(t),z(t)) & Q_R(e(t)) & \frac{d}{dt}Q_R(e(t))
& \min_{0\le t\le1} Q_R(e(t))\\
\hline
p_0p_1
&
\left(\frac t2,\frac t2\right)
&
-\frac{2219t^2-3705t-68}{2}
&
-\frac{4438t-3705}{2}
&
34\text{ at }t=0
\\[2pt]
p_1p_2
&
\left(\frac12-\frac t4,\frac12\right)
&
-\frac{2325t^2-733t-6216}{8}
&
-\frac{4650t-733}{8}
&
578\text{ at }t=1
\\[2pt]
p_2p_0
&
\left(\frac{1-t}{4},\frac{1-t}{2}\right)
&
-\frac{6763t^2-2411t-4624}{8}
&
-\frac{13526t-2411}{8}
&
34\text{ at }t=1
\end{array}
\]
Therefore \(\min_R Q_R=34\), and hence
\[
    m_S=\frac{34}{10000}=\frac{17}{5000}.
\]

\paragraph*{Case 2: \(U_{AA}\cap V_X\).}
The region has vertices
\[
    p_0=\left(\frac14,\frac12\right),\quad
    p_1=\left(\frac12,\frac12\right),\quad
    p_2=\left(\frac12,\frac34\right),\quad
    p_3=\left(\frac13,\frac23\right).
\]
On this region,
\[
    Q_R(x,z)
    =
    -4650x^2+4862xz+\frac{5105}{2}x
    -4650z^2+\frac{3705}{2}z-316.
\]
The stationary equations are
\[
    -9300x+4862z+\frac{5105}{2}=0,
    \qquad
    4862x-9300z+\frac{3705}{2}=0.
\]
Solving gives
\[
    (x,z)=
    \left(
        \frac{32745105}{62850956},
        \frac{29638505}{62850956}
    \right),
\]
which is infeasible, since it violates \(x\le z\) and also has
\(z<\frac12\). We therefore only need to check the boundary:
\[
\renewcommand{\arraystretch}{1.25}
\begin{array}{c|c|c|c|c}
\text{edge} & (x(t),z(t)) & Q_R(e(t)) & \frac{d}{dt}Q_R(e(t))
& \min_{0\le t\le1} Q_R(e(t))\\
\hline
p_0p_1
&
\left(\frac14+\frac t4,\frac12\right)
&
-\frac{2325t^2-5317t-3224}{8}
&
-\frac{4650t-5317}{8}
&
403\text{ at }t=0
\\[2pt]
p_1p_2
&
\left(\frac12,\frac12+\frac t4\right)
&
-\frac{2325t^2+733t-6216}{8}
&
-\frac{4650t+733}{8}
&
\frac{1579}{4}\text{ at }t=1
\\[2pt]
p_2p_3
&
\left(\frac12-\frac t6,\frac34-\frac{t}{12}\right)
&
-\frac{6763t^2+2439t-28422}{72}
&
-\frac{13526t+2439}{72}
&
\frac{4805}{18}\text{ at }t=1
\\[2pt]
p_3p_0
&
\left(\frac13-\frac{t}{12},\frac23-\frac t6\right)
&
-\frac{6763t^2-16559t-19220}{72}
&
-\frac{13526t-16559}{72}
&
\frac{4805}{18}\text{ at }t=0
\end{array}
\]
Therefore
\[
    \min_R Q_R=\frac{4805}{18},
\]
attained at \(p_3=\left(\frac13,\frac23\right)\). Hence
\[
    m_S
    =
    \frac{1}{10000}\cdot\frac{4805}{18}
    =
    \frac{961}{36000}.
\]

\paragraph*{Case 3: \(U_{AA}\cap V_R\).}
The region has vertices
\[
    p_0=\left(\frac12,\frac12\right),\qquad
    p_1=(1,1),\qquad
    p_2=\left(\frac12,\frac34\right).
\]
On this region,
\[
    Q_R(x,z)
    =
    -4650x^2+4862xz+\frac{5105}{2}x
    -4650z^2+\frac{5105}{2}z-666.
\]
The stationary equations are
\[
    -9300x+4862z+\frac{5105}{2}=0,
    \qquad
    4862x-9300z+\frac{5105}{2}=0.
\]
Solving gives
\[
    (x,z)=\left(\frac{5105}{8876},\frac{5105}{8876}\right),
\]
which is feasible. Its value is
\[
    Q_R\left(\frac{5105}{8876},\frac{5105}{8876}\right)
    =
    \frac{14238193}{17752}
    >
    1.
\]
The edge checks are
\[
\renewcommand{\arraystretch}{1.25}
\begin{array}{c|c|c|c|c}
\text{edge} & (x(t),z(t)) & Q_R(e(t)) & \frac{d}{dt}Q_R(e(t))
& \min_{0\le t\le1} Q_R(e(t))\\
\hline
p_0p_1
&
\left(\frac12+\frac t2,\frac12+\frac t2\right)
&
-\frac{2219t^2-667t-1554}{2}
&
-\frac{4438t-667}{2}
&
1\text{ at }t=1
\\[2pt]
p_1p_2
&
\left(1-\frac t2,1-\frac t4\right)
&
-\frac{6763t^2-11313t-8}{8}
&
-\frac{13526t-11313}{8}
&
1\text{ at }t=0
\\[2pt]
p_2p_0
&
\left(\frac12,\frac34-\frac t4\right)
&
-\frac{2325t^2-3983t-4558}{8}
&
-\frac{4650t-3983}{8}
&
\frac{2279}{4}\text{ at }t=0
\end{array}
\]
Therefore \(\min_R Q_R=1\), attained at \((1,1)\), and hence
\[
    m_S=\frac{1}{10000}.
\]

\paragraph*{Case 4: \(U_{AB}\cap V_L\).}
The region has vertices
\[
    p_0=(0,0),\qquad
    p_1=\left(\frac14,\frac12\right),\qquad
    p_2=\left(0,\frac12\right).
\]
On this region,
\[
    Q_R(x,z)
    =
    4650x^2-4438xz+\frac{3705}{2}x
    -2325z^2+\frac{3705}{2}z+34.
\]
The stationary equations are
\[
    9300x-4438z+\frac{3705}{2}=0,
    \qquad
    -4438x-4650z+\frac{3705}{2}=0.
\]
Solving gives
\[
    (x,z)=
    \left(
        -\frac{10335}{1656338},
        \frac{1339455}{3312676}
    \right),
\]
which is infeasible because the \(x\)-coordinate is negative. The edge
checks are
\[
\renewcommand{\arraystretch}{1.25}
\begin{array}{c|c|c|c|c}
\text{edge} & (x(t),z(t)) & Q_R(e(t)) & \frac{d}{dt}Q_R(e(t))
& \min_{0\le t\le1} Q_R(e(t))\\
\hline
p_0p_1
&
\left(\frac t4,\frac t2\right)
&
-\frac{6763t^2-11115t-272}{8}
&
-\frac{13526t-11115}{8}
&
34\text{ at }t=0
\\[2pt]
p_1p_2
&
\left(\frac{1-t}{4},\frac12\right)
&
\frac{2325t^2-3917t+4624}{8}
&
\frac{4650t-3917}{8}
&
\frac{27660311}{74400}\text{ at }t=\frac{3917}{4650}
\\[2pt]
p_2p_0
&
\left(0,\frac{1-t}{2}\right)
&
-\frac{2325t^2-945t-1516}{4}
&
-\frac{15(310t-63)}{4}
&
34\text{ at }t=1
\end{array}
\]
Thus \(\min_R Q_R=34\), attained at \((0,0)\), and therefore
\[
    m_S=\frac{17}{5000}.
\]

\paragraph*{Case 5: \(U_{AB}\cap V_X\).}
The region has vertices
\[
    p_0=\left(0,\frac12\right),\qquad
    p_1=\left(\frac14,\frac12\right),\qquad
    p_2=\left(\frac13,\frac23\right).
\]
On this region,
\[
    Q_R(x,z)
    =
    4650x^2-4438xz+\frac{5105}{2}x
    -2325z^2+\frac{3705}{2}z-316.
\]
The stationary equations are
\[
    9300x-4438z+\frac{5105}{2}=0,
    \qquad
    -4438x-4650z+\frac{3705}{2}=0.
\]
Solving gives
\[
    (x,z)=
    \left(
        -\frac{1823865}{31470422},
        \frac{28556245}{62940844}
    \right),
\]
which is infeasible because \(x<0\) and \(z<\frac12\). The edge checks are
\[
\renewcommand{\arraystretch}{1.25}
\begin{array}{c|c|c|c|c}
\text{edge} & (x(t),z(t)) & Q_R(e(t)) & \frac{d}{dt}Q_R(e(t))
& \min_{0\le t\le1} Q_R(e(t))\\
\hline
p_0p_1
&
\left(\frac t4,\frac12\right)
&
\frac{2325t^2+667t+232}{8}
&
\frac{4650t+667}{8}
&
29\text{ at }t=0
\\[2pt]
p_1p_2
&
\left(\frac14+\frac{t}{12},\frac12+\frac t6\right)
&
-\frac{6763t^2+3033t-29016}{72}
&
-\frac{13526t+3033}{72}
&
\frac{4805}{18}\text{ at }t=1
\\[2pt]
p_2p_0
&
\left(\frac{1-t}{3},\frac23-\frac t6\right)
&
\frac{7399t^2-15965t+9610}{36}
&
\frac{14798t-15965}{36}
&
29\text{ at }t=1
\end{array}
\]
Therefore \(\min_R Q_R=29\), attained at
\(\left(0,\frac12\right)\), and
\[
    m_S=\frac{29}{10000}.
\]

\paragraph*{Case 6: \(U_{CA}\cap V_X\).}
The region has vertices
\[
    p_0=\left(\frac13,\frac23\right),\qquad
    p_1=\left(\frac12,\frac34\right),\qquad
    p_2=\left(\frac12,1\right).
\]
On this region,
\[
    Q_R(x,z)
    =
    -2325x^2-4438xz+\frac{14405}{2}x
    +4650z^2-\frac{14895}{2}z+2009.
\]
The stationary equations are
\[
    -4650x-4438z+\frac{14405}{2}=0,
    \qquad
    -4438x+9300z-\frac{14895}{2}=0.
\]
Solving gives
\[
    (x,z)=
    \left(
        \frac{1785855}{3312676},
        \frac{1752515}{1656338}
    \right),
\]
which is infeasible because \(x>\frac12\) and \(z>1\). The boundary
checks are
\[
\renewcommand{\arraystretch}{1.25}
\begin{array}{c|c|c|c|c}
\text{edge} & (x(t),z(t)) & Q_R(e(t)) & \frac{d}{dt}Q_R(e(t))
& \min_{0\le t\le1} Q_R(e(t))\\
\hline
p_0p_1
&
\left(\frac13+\frac t6,\frac23+\frac{t}{12}\right)
&
-\frac{6763t^2-15965t-19220}{72}
&
-\frac{13526t-15965}{72}
&
\frac{4805}{18}\text{ at }t=0
\\[2pt]
p_1p_2
&
\left(\frac12,\frac34+\frac t4\right)
&
\frac{2325t^2-5383t+3158}{8}
&
\frac{4650t-5383}{8}
&
\frac{25}{2}\text{ at }t=1
\\[2pt]
p_2p_0
&
\left(\frac12-\frac t6,1-\frac t3\right)
&
\frac{7399t^2+1761t+450}{36}
&
\frac{14798t+1761}{36}
&
\frac{25}{2}\text{ at }t=0
\end{array}
\]
Thus \(\min_R Q_R=\frac{25}{2}\), attained at
\(\left(\frac12,1\right)\), and
\[
    m_S
    =
    \frac{1}{10000}\cdot\frac{25}{2}
    =
    \frac{1}{800}.
\]

\paragraph*{Case 7: \(U_{CA}\cap V_R\).}
The region has vertices
\[
    p_0=\left(\frac12,\frac34\right),\qquad
    p_1=(1,1),\qquad
    p_2=\left(\frac12,1\right).
\]
On this region,
\[
    Q_R(x,z)
    =
    -2325x^2-4438xz+\frac{14405}{2}x
    +4650z^2-\frac{13495}{2}z+1659.
\]
The stationary equations are
\[
    -4650x-4438z+\frac{14405}{2}=0,
    \qquad
    -4438x+9300z-\frac{13495}{2}=0.
\]
Solving gives
\[
    (x,z)=
    \left(
        \frac{37037845}{62940844},
        \frac{31670285}{31470422}
    \right),
\]
which is infeasible because \(z>1\). The edge checks are
\[
\renewcommand{\arraystretch}{1.25}
\begin{array}{c|c|c|c|c}
\text{edge} & (x(t),z(t)) & Q_R(e(t)) & \frac{d}{dt}Q_R(e(t))
& \min_{0\le t\le1} Q_R(e(t))\\
\hline
p_0p_1
&
\left(\frac12+\frac t2,\frac34+\frac t4\right)
&
-\frac{6763t^2-2213t-4558}{8}
&
-\frac{13526t-2213}{8}
&
1\text{ at }t=1
\\[2pt]
p_1p_2
&
\left(1-\frac t2,1\right)
&
-\frac{2325t^2-3771t-4}{4}
&
-\frac{3(1550t-1257)}{4}
&
1\text{ at }t=0
\\[2pt]
p_2p_0
&
\left(\frac12,1-\frac t4\right)
&
\frac{2325t^2-667t+2900}{8}
&
\frac{4650t-667}{8}
&
\frac{26525111}{74400}\text{ at }t=\frac{667}{4650}
\end{array}
\]
Therefore \(\min_R Q_R=1\), attained at \((1,1)\), and hence
\[
    m_S=\frac{1}{10000}.
\]

\paragraph*{Case 8: \(U_{CB}\cap V_L\).}
This intersection is lower-dimensional. Indeed, \(z\le\frac12\) and
\(z\ge\frac{1+x}{2}\) imply \(1+x\le2z\le1\), so \(x=0\) and
\(z=\frac12\). Hence
\[
    U_{CB}\cap V_L
    =
    \left\{\left(0,\frac12\right)\right\}.
\]
On this region,
\[
    Q_R(x,z)
    =
    6975x^2-13738xz+\frac{13005}{2}x
    +6975z^2-\frac{14895}{2}z+2359.
\]
Since the region consists of a single point, the minimum is obtained by
direct evaluation:
\[
\begin{aligned}
    Q_R\left(0,\frac12\right)
    &=
    6975\cdot0^2
    -13738\cdot0\cdot\frac12
    +\frac{13005}{2}\cdot0
    +6975\cdot\frac14
    -\frac{14895}{2}\cdot\frac12
    +2359  \\
    &=
    \frac{6975}{4}-\frac{14895}{4}+2359
    =
    379.
\end{aligned}
\]
Therefore
\[
    m_S=\frac{379}{10000}.
\]

\paragraph*{Case 9: \(U_{CB}\cap V_X\).}
The region has vertices
\[
    p_0=\left(0,\frac12\right),\qquad
    p_1=\left(\frac13,\frac23\right),\qquad
    p_2=\left(\frac12,1\right),\qquad
    p_3=(0,1).
\]
On this region,
\[
    Q_R(x,z)
    =
    6975x^2-13738xz+\frac{14405}{2}x
    +6975z^2-\frac{14895}{2}z+2009.
\]
The stationary equations are
\[
    13950x-13738z+\frac{14405}{2}=0,
    \qquad
    -13738x+13950z-\frac{14895}{2}=0.
\]
Equivalently,
\[
    \begin{pmatrix}
        13950 & -13738\\
        -13738 & 13950
    \end{pmatrix}
    \binom{x}{z}
    =
    \binom{-14405/2}{14895/2}.
\]
Solving this system gives
\[
    (x^*,z^*)
    =
    \left(
        \frac{57465}{183433},
        \frac{618085}{733732}
    \right).
\]
This point is feasible. Indeed,
\[
    \frac12-x^*=\frac{68503}{366866}\ge0,\qquad
    z^*-\frac12=\frac{251219}{733732}\ge0,\qquad
    1-z^*=\frac{115647}{733732}\ge0,
\]
and the two defining inequalities of \(U_{CB}\) hold because
\[
    z^*-2x^*
    =
    \frac{158365}{733732}\ge0,
    \qquad
    z^*-\frac{1+x^*}{2}
    =
    \frac{136289}{733732}\ge0.
\]
Moreover, the Hessian matrix of \(Q_R\) is
\[
    \begin{pmatrix}
        13950 & -13738\\
        -13738 & 13950
    \end{pmatrix},
\]
whose eigenvalues are \(13950-13738=212\) and
\(13950+13738=27688\). Hence \(Q_R\) is strictly convex on this region.
Therefore this feasible stationary point is the unique global minimizer
over \(R\). Its value is
\[
    Q_R(x^*,z^*)
    =
    \frac{1027577}{2934928}.
\]
For completeness, we also record the boundary checks:
\[
\renewcommand{\arraystretch}{1.25}
\begin{array}{c|c|c|c|c}
\text{edge} & (x(t),z(t)) & Q_R(e(t)) & \frac{d}{dt}Q_R(e(t))
& \min_{0\le t\le1} Q_R(e(t))\\
\hline
p_0p_1
&
\left(\frac t3,\frac12+\frac t6\right)
&
\frac{7399t^2+1167t+1044}{36}
&
\frac{14798t+1167}{36}
&
29\text{ at }t=0
\\[2pt]
p_1p_2
&
\left(\frac13+\frac t6,\frac23+\frac t3\right)
&
\frac{7399t^2-16559t+9610}{36}
&
\frac{14798t-16559}{36}
&
\frac{25}{2}\text{ at }t=1
\\[2pt]
p_2p_3
&
\left(\frac{1-t}{2},1\right)
&
\frac{6975t^2-879t+50}{4}
&
\frac{3(4650t-293)}{4}
&
\frac{69151}{12400}\text{ at }t=\frac{293}{4650}
\\[2pt]
p_3p_0
&
\left(0,1-\frac t2\right)
&
\frac{6975t^2-13005t+6146}{4}
&
\frac{45(310t-289)}{4}
&
\frac{10415}{496}\text{ at }t=\frac{289}{310}
\end{array}
\]
All boundary minima are larger than
\(\frac{1027577}{2934928}\). Hence
\[
    \min_R Q_R
    =
    \frac{1027577}{2934928},
\]
and therefore
\[
    m_S
    =
    \frac{1}{10000}\cdot
    \frac{1027577}{2934928}
    =
    \frac{1027577}{29349280000}.
\]

Combining all the cases above, every nonempty intersection has
a nonnegative lower bound for \(D\). Therefore \(D(x,z)\ge0\) on the triangle
\(0\le x\le z\le1\). By symmetry of \(D\), the same inequality holds on
the whole square \([0,1]^2\).

\end{document}